\documentclass[a4paper,titlepage]{article}

\usepackage[utf8]{inputenc}
\usepackage[english]{babel}
\usepackage{amsfonts,amsthm,amsmath,amssymb,enumerate, amsopn,mmap,afterpage,verbatim}
\usepackage{graphicx}
\usepackage{textcomp}

\theoremstyle{plain}
\newtheorem{lemma}{Lemma}
\newtheorem{proposition}{Proposition}
\newtheorem{theorem}{Theorem}
\newtheorem{corollary}{Corollary}
\newtheorem*{qdet}{$Q$-determinant}

\theoremstyle{definition}
\newtheorem{definition}{Definition}

\theoremstyle{remark}
\newtheorem{remark}{Remark}
\newtheorem*{requirements}{Requirements}
\newtheorem{example}{Example}

\title{Parallelism Resource of Numerical Algorithms.\\Version 1}

\author{Valentina N. Aleeva, Rifkhat~Zh.~Aleev}

\begin{document}

\maketitle

\begin{abstract}
The paper is devoted to an approach to solving a problem of the efficiency of parallel computing.
The theoretical basis of this approach is the concept of a $Q$-determinant.
Any numerical algorithm has a $Q$-determinant.
The $Q$-determinant of the algorithm has clear structure and is convenient for implementation.
The $Q$-determinant consists of $Q$-terms.
Their number is equal to the number of output data items.
Each $Q$-term describes all possible ways to compute one of the output data items based on the input data.

We also describe a software $Q$-system for studying the parallelism resource of numerical algorithms.
This system enables to compute and compare the parallelism resources of numerical algorithms.
The application of the $Q$-system is shown on the example of numerical algorithms with different structures of $Q$-determinants.
Furthermore, we suggest a method for designing of parallel programs for numerical algorithms.
This method is based on a representation of a numerical algorithm in the form of a $Q$-determinant.
As a result, we can obtain the program using the parallelism resource of the algorithm completely.
Such programs are called $Q$-effective.

The results of this research can be applied to increase the implementation efficiency of numerical algorithms, methods, as well as algorithmic problems on parallel computing systems.
\end{abstract}


Additional Key Words and Phrases: \emph{$Q$-term of algorithm, $Q$-determinant of algorithm, representation of
algorithm in form of $Q$-determinant, $Q$-effective implementation of algorithm, parallelism resource of algorithm,
software $Q$-system, $Q$-effective program, $Q$-effective programming}

\renewcommand{\topfraction}{1.0}

\section{Introduction}

There is a considerable difference in the computational power of parallel computing systems and its use.
The existence of this fact is of great importance for parallel computing.
One of the reasons for the above difference is an inadequate implementation of algorithms on parallel computing systems.
In particular, it can be if the parallelism resource of the algorithm is used incompletely.
So, the computing resources of a parallel computing system can not be used enough when implementing the algorithm.

We will give a brief overview of some researches of the parallelism resource of numerical algorithms and its implementation.

\textbf{First,} we note \cite{al:vv,al:voev} where there is a very important and developed research of the parallel structure of algorithms and programs for their implementation on parallel computing systems.
These papers contain definitions and studies of the graphs of algorithms.
These researches are adapted in the open encyclopedia AlgoWiki \cite{al:ant,al:adv}.
However, the papers using these studies do not consider any software for studying the parallelism resource of algorithms.

\textbf{Second,} we note there are proposed several approaches to the development of parallel programs.
This led to the creation of various parallel programming languages and other tools.
The \textsf{T-system} \cite{al:tsystem} is one of these developments.
It provides a programming environment with support for automatic dynamic parallelization of programs.
However, it cannot be asserted that the creation of parallel programs using the \textsf{T-system} makes full use of the parallelism resource of the algorithm.
The parallel program synthesis is another approach to creating parallel programs.
This approach is to construct new parallel algorithms using the knowledge base of parallel algorithms to solve more complex problems.
The technology of fragmented programming, its implementation language, and the programming system \textsf{LuNA} are developed on the basis of the parallel programming synthesis method \cite{al:mal}.
This approach does not solve the problem of research and use of the parallelism resource of algorithm, despite the fact that it is universal.
To overcome resource limitations, the author of the paper \cite{al:legalov} suggests methods for constructing parallel programs using a functional programming language independent of computer architecture.
However, there isn't shown that the created programs use the entire parallelism resource of algorithms.

\textbf{Third.}
There are many studies on the development of parallel programs that take into account the specifics of algorithms and architecture of parallel computing systems.
Examples of such studies are \cite{al:wang,al:li,al:suplatov,al:prifti,al:you,al:matv,al:set}.
These studies improve the efficiency of implementing specific algorithms or implementing algorithms on parallel computing systems of a particular architecture.
However, they don't provide a general universal approach.

\textbf{Fourth.}
Perhaps the above review is not complete.
However, we have previously noted that the parallelism resource of algorithms by realization on parallel computing systems is often not used completely.
So, it appears that there is currently no solution to the problem for the research and use of the parallelism resource of algorithms.
Therefore, the results of this paper can be considered as one of the solutions to this problem.

The concept of a $Q$-determinant is the theoretical basis of the research of this paper.
We describe the development of a software system called the $Q$-system to research the parallelism resource of numerical algorithms.
We will also describe how to develop a program that uses the parallelism resource of the numerical algorithm completely.
In addition, we suggest a programming technology called $Q$-effective programming to improve the efficiency of parallel computing based on the results obtained.
This paper continues and summarizes the results of studies presented in \cite{al:al,al:alss,al:ruscdays18,al:glosic, al:ruscdays19}.

\section{The concept of a $Q$-determinant}\label{s:Qdet}

We describe a mathematical model of the concept of a $Q$-determinant.

\subsection{Expressions}

Let $B=\{b_1,b_2,\dots\}$ be a finite or countable set of variables, and $Q$ be a finite set of operations.
Suppose that \emph{all} operations of $Q$ \emph{are $0$ary (constants), unary, or binary.}
For example,
\begin{align*}
Q=\{&+,-,\cdot,/ \text{ (arithmetical operations)},\\
        &\vee,\wedge,\neg\text{ (logical operations)},\\
        &=,<,\leq,>,\geq,\neq\text{ (comparison operations)}\}.
\end{align*}

Every expression $w$ has the \emph{nesting level} $T^w$.
\begin{definition}
By induction, we define \emph{the expression, its nesting level \emph{and} its subexpressions}.
We also relate the nesting level and operations.

\begin{enumerate}
\item
The constants and elements of the set $B$ are \emph{expressions} and have \emph{zero nesting level}.
\item
If $w$ is an expression, then $(w)$ is an \emph{expression} and $T^w=T^{(w)}$ also.
\item
Let $w$ be an expression, $T^w=i-1$ ($i\geq1$) and $f\in Q$ is an unary operation.
Then $f(w)$ is an \emph{expression} and $T^{f(w)}=i$.
We call $w$ \emph{the subexpression of the $(i-1)$th nesting level of the expression $f(w)$} and $f$ \emph{the operation of the $i$th nesting level}.
\item
Let $w$ and $v$ be an expressions, $T^w=i$, $T^v=j$ and $g\in Q$ be a binary operation.
Then $g(w, v)$ is an \emph{expression} and $T^{g(w,v)}=k$, where $k=\max\{i,j\}+1$.
We call $w$ and $v$ \emph{the subexpressions of the expression $g(w,v)$} with nesting levels $i$ and $j$, respectively, and $g$ \emph{the operation of the $k$th nesting level}.
\end{enumerate}
\end{definition}

\begin{example}
We point out some expressions and their nesting levels.
\begin{enumerate}\label{ex:exp}
\item
$w_1=b_1\cdot(b_2+b_3)/b_4$ and $T^{w_1}=3$;
\item
$w_2=((b_1+b_2)\leq(b_3\cdot b_4))\vee\neg(b_5\leq b_6)$ and $T^{w_2}=3$;
\item
$w_3=((b_1\geq b_3)\wedge((b_2-b_4)\neq 0))\wedge(b_5=0)$ and $T^{w_3}=4$.
\end{enumerate}
\end{example}

\begin{definition}
We call an expression \emph{a chain of length $n$} if it is the result of some associative operation from $Q$ on expressions whose number is $n$.
As usual, we can write a chain without parentheses.
\end{definition}

\begin{example}
Examples of chains are the following:
\begin{enumerate}
\item
$b_1+b_2+b_3+b_4$ is a chain of length 4;
\item
$(b_1+b_2)\cdot(b_4-b_7)\cdot(b_2+b_4)$ is a chain of length 3;
\item
$(b_1\leq b_2)\vee(b_3\geq b_5)\vee\neg(b_2\leq b_4)$ is a chain of length 3.
\end{enumerate}
\end{example}

We interpret the expressions into the real number field $\mathbb{R}$.

\begin{definition}
Let $b_i\in B$.
Then the assignment of the variable $b_i$ of a specific value from $\mathbb{R}$ is called \emph{the interpretation of the variable} $b_i$.
\end{definition}

\begin{definition}
We say that \emph{the interpretation of the expression is specified} if the interpretation of all variables in the expression is specified.
\end{definition}

If the interpretation of the expression is specified, then we can find the value of the expression.
\begin{example}
As an example, to find the values of the expressions, consider the expressions from Example \ref{ex:exp}.
First, we interpret the variables
\[
b_i=i\text{ for every }i\in\{1,2,\dots,6\}.
\]
\begin{enumerate}\label{ex:valexp}
\item
We have
\[
w_1=b_1\cdot(b_2+b_3)/b_4=1\cdot(2+3)/4=5/4.
\]
So, the value of $w_1$, under this interpretation, is $5/4$.
\item
We get
\begin{align*}
w_2&=((b_1+b_2)\leq(b_3\cdot b_4))\vee\neg(b_5\leq b_6)=\\
&=((1+2)\leq(3\cdot4))\vee\neg(4\leq6)=\\
&=(3\leq12)\vee\neg(4\leq6).
\end{align*}
The value of $w_2$ is \textsf{true}.
\item
Finally,
\begin{align*}
w_3&=((b_1\geq b_3)\wedge((b_2-b_4)\neq 0))\wedge(b_5=0)=\\
&=((1\geq3)\wedge((2-4)\neq 0))\wedge(5=0)=\\
&=((1\geq3)\wedge(-2\neq 0))\wedge(5=0).
\end{align*}
Therefore, the value of $w_3$ is \textsf{false}.
\end{enumerate}
\end{example}

 \subsection{$Q$-terms and their values}

Often some additional data have an influence on expressions.
We call this data parameters.
More exactly, let's define the notion of parameters.
In the standard sense of mathematical logic \cite[Section 16]{al:ersh}, a set of parameters is a set of free variables of expression, cf. \cite[Section 1.2]{al:gj}.
So, we clarify our idea of parameters.

\begin{definition}
We keep the following agreements.
\begin{enumerate}
\item
Let $N$ be a set of parameters.
Then $N=\emptyset$ or $N=\{n_1,\dots,n_k\}$, where $k\geq1$, and $n_i$ is equal to any positive integer for every $i\in\{1,\dots,k\}$.
\item
If $N=\{n_1,\dots,n_k\}$, then as
\[
\bar{N}=(\bar {n}_1,\dots,\bar{n}_k),
\]
we denote the $k$-tuple, where $\bar {n}_i$ is some given value of the parameter $n_i$ for every $i\in\{1,\dots,k\}$.
\item
By $\left\{\bar{N}\right\}$ we denote the set of all possible $k$-tuples $\bar{N}$.
\end{enumerate}
\end{definition}

Now we introduce the concept of an unconditional $Q$-term.
\begin{definition}
If $N=\emptyset$, then we say that every expression $w$ over $B$ and $Q$ is an \emph{unconditional $Q$-term}.

Let $N\neq\emptyset$ and $V$ be a set of all expressions over $B$ and $Q$.
Suppose that we have a map $w:\left\{\bar{N}\right\}\to V\cup\emptyset$.
Then this map $w$ is called an \emph{unconditional $Q$-term}.
\end{definition}

Thus, for $N =\emptyset$ the concept of an unconditional $Q$-term and an expression over $B$ and $Q$ coincide.
If $N\neq\emptyset$, then for every $\bar{N}\in\left\{\bar{N}\right\}$ we have $w(\bar{N})$ is either some expression over $B$ and $Q$, or $w(\bar{N})=\emptyset$, meaning that $w(\bar{N})$ is undefined.

\begin{example}
Examples of unconditional $Q$-terms are
\begin{enumerate}
\item
$b_1+b_2\cdot b_3-b_4$, here $N=\emptyset$;
\item
$|b_1+b_2+\dots +b_n|$, here $N=\{n\}$;
\item
$(b_1=0)\wedge(b_3=0)\wedge\dots\wedge(b_{1+2\cdot n_1}=0)\wedge((b_2\cdot b_4\cdot\cdots\cdot b_{2+2\cdot n_2})>0)$, and here $N=\{n_1,n_2\}$.
\end{enumerate}
\end{example}

\begin{definition}
If $N =\emptyset$, then finding the value of an expression $w$ is \emph{finding the value of an unconditional $Q$-term $w$ under any interpretation of the variables of $B$}.

If $N\neq\emptyset$ and $w(\bar{N})\neq\emptyset$, then $w(\bar{N})$ is an expression over $B$ and $Q$.
We can find the value of the expression $w(\bar{N})$.
Certainly, we omit the value of $w(\bar{N})=\emptyset$.
Hence, we have \emph{finding the value of an unconditional $Q$-term $w$ under any interpretation of the variables of $B$}.
\end{definition}

\begin{definition}
If $N =\emptyset$, then the nesting level of the expression $w$ is \emph{the nesting level $T^w$ of the unconditional $Q$-term $w$}.

If $N\neq\emptyset$, then \emph{the nesting level of the unconditional $Q$-term $w$} is the partial function $T^w :\left\{\bar{N}\right\}\to T^{w(\bar{N})}$,
where $T^{w(\bar{N})}$ is the nesting level of the expression $w(\bar{N})$ for $w(\bar{N})\neq\emptyset$.
Also, we don't define the nesting level of the unconditional $Q$-term $w$, if $w(\bar{N})=\emptyset$.
\end{definition}

\begin{definition}
Let $N = \emptyset$ and $w$ be an unconditional $Q$-term.
Suppose that the expression $w$ over $B$ and $Q$ has a value of a logical type under any interpretation of the variables of $B$.

Then the unconditional $Q$-term $w$ is called the \emph{unconditional logical $Q$-term}.

Let $N\neq\emptyset$ and $w$ be an unconditional $Q$-term.
If the expression $w(\bar{N})$ for every $\bar{N}\in\{\bar{N}\}$ has a value of a logical type under any interpretation of the variables of $B$,
then the unconditional $Q$-term $w$ is called the \emph{unconditional logical $Q$-term}.
\end{definition}

\begin{definition}
Let $u_1,\dots,u_l$ be unconditional logical $Q$-terms, $w_1,\dots,w_l$ are unconditional $Q$-terms.
We denote
\[
\left(\widehat{u},\widehat{w}\right)=\left\{(u_i,w_i)\right\}_{i\in\{1,\dots,l\}}
\]
and call a \emph{conditional $Q$-term of length $l$}.
\end{definition}

We describe finding the value of a conditional $Q$-term $(\widehat{u},\widehat{w})$ under the interpretation of the variables of $B$.
\begin{definition}
Let $N = \emptyset$.
We find the values of the expressions $u_i,w_i$ for $i\in\{1,\dots,l\}$.
Under this finding of the values we can find a pair $u_{i_0},w_{i_0}$ such that $u_{i_0}$ has the value \textsf{true}.
Therefore, we can find the value of $w_{i_0}$.
Then we suppose that $\left(\widehat{u},\widehat{w}\right)$ has the value $w_{i_0}$.
Otherwise, we suppose that the value of $\left(\widehat{u},\widehat{w}\right)$ under the interpretation of the variables of $B$ is not determined.

Let  $N\neq\emptyset$ and $\bar{N}\in\left\{\bar{N}\right\}$.
We find the expressions $u_i(\bar{N}),w_i(\bar{N})$ for $i\in\{1,\dots,l\}$.
Under this finding of the values we can find a pair $u_{i_0}(\bar{N}),w_{i_0}(\bar{N})$ such that $u_{i_0}(\bar{N})$ has the value \textsf{true}.
Therefore, we can find the value of $w_{i_0}(\bar{N})$.
Then we suppose that $\left(\widehat{u},\widehat{w}\right)$ has the value $w_{i_0}(\bar{N})$.
Otherwise, we suppose that the value of $\left(\widehat{u},\widehat{w}\right)$ for $\bar{N}$ and under this interpretation of the variables of $B$ is not determined.
\end{definition}

\begin{definition}
Let $\left(\widehat{u},\widehat{w}\right)=\left\{(u_i,w_i)\right\}_{i\in\{1,2,\dots\}}$ be a countable set of pairs of unconditional $Q$-terms.
Assume that $\left\{(u_i,w_i)\right\}_{i\in\{1,\dots,l\}}$ is a conditional $Q$-term for any $l<\infty$.
Then we call $\left(\widehat{u},\widehat{w}\right)$ a \emph{conditional infinite $Q$-term}.
\end{definition}

Finally, we determine the value of a conditional infinite $Q$-term $\left(\widehat{u},\widehat{w}\right)$ under the interpretation of the variables of $B$.
\begin{definition}
Let $N = \emptyset$.
First of all, we find the values of the expressions $u_i,w_i$ for $i\in\{1,2,\dots\}$.
Under this finding of the values we can find a pair $u_{i_0},w_{i_0}$ such that $u_{i_0}$ has the value \textsf{true}.
Therefore, we can find the value of $w_{i_0}$.
Then we suppose that $\left(\widehat{u},\widehat{w}\right)$ has the value $w_{i_0}$.
Otherwise, we suppose that the value of $\left(\widehat{u},\widehat{w}\right)$ under the interpretation of the variables of $B$ is not determined.

Let  $N\neq\emptyset$ and $\bar{N}\in\left\{\bar{N}\right\}$.
First of all, we find the values of the expressions $u_i(\bar{N}),w_i(\bar{N})$ for $i\in\{1,2,\dots\}$.
Under this finding of the values we can find a pair $u_{i_0}(\bar{N}),w_{i_0}(\bar{N})$ such that $u_{i_0}(\bar{N})$ has the value \textsf{true}.
Therefore, we can find the value of $w_{i_0}(\bar{N})$.
Then we suppose that $\left(\widehat{u},\widehat{w}\right)$  has the value $w_{i_0}(\bar{N})$.
Otherwise, we suppose the value of $\left(\widehat{u},\widehat{w}\right)$ for $\bar{N}$ and under the interpretation of the variables of $B$ is not determined.
\end{definition}

\begin{remark}
If it does not matter whether a $Q$-term is unconditional, conditional or conditional infinite, then we call it \emph{a $Q$-term}.
\end{remark}

 \subsection{The concept of a $Q$-determinant of an algorithm}\label{ss:bmodel}
 
Consider an algorithmic problem
\begin{equation}
\vec{y}=F(N,B), \label{ap}
\end{equation}
where $N$ is a set of dimension parameters of the problem, $B$ is a set of input data, $\vec{y}=(y_1,\dots,y_m)$ is a set of output data,
$y_i\notin B$ for every $i\in\{1,\dots,m\}$, the integer $m$ is either a constant or the value of a computable parameter function $N$ under the condition
$N\neq\emptyset$, cf. \cite[p.~4]{al:gj} and \cite[p.~7--8]{al:knuth}.

We introduce the concept of a $Q$-determinant of an algorithm.
\begin{definition}
Let $\mathcal{A}$ be an numerical algorithm for solving an algorithmic problem $\vec{y}=F(N,B)$ and $M=\{1,\dots,m\}$.

Suppose that the algorithm $\mathcal{A}$ consists in finding for every $i\in M$ the value of $y_i$ when the value of a $Q$-term $f_i$ is found.
Then the set of $Q$-terms
\begin{equation}
\left\{f_i\mid i\in M\right\} \label{qd}
\end{equation}
is called \emph{the $Q$-determinant of the algorithm $\mathcal{A}$}.
Also a system of equations
\begin{equation}
y_i=f_i\text{ for all }i\in M\label{fqd}
\end{equation}
is called \emph{a representation of the algorithm $\mathcal{A}$ in the form of a $Q$-determinant}.
\end{definition}

\subsection{The concept of the $Q$-effective implementation of an algorithm}

It is very important how an algorithm is computed.
\begin{definition}
Let the algorithm $\mathcal{A}$ be represented in the form of a $Q$-deter\-minant $y_i=f_i$ for all $i\in M$ (cf. (\ref{fqd})).
The process of computing the $Q$-terms $f_i$ for all $i\in M$ is called an \emph{implementation of the algorithm} $\mathcal{A}$.
If an implementation of the algorithm is such that two or more operations are performed simultaneously, then it will be called a \emph{parallel implementation}.
\end{definition}

We describe a very important implementation of the algorithm $\mathcal{A}$.

For that we need some partition of the set $M$.
\begin{definition}
More exactly, suppose that $U$, $C$ and $I$ form a partition of the set $M =\{1,\dots,m\}$ with empty terms, that is:
\begin{enumerate}
\item
$U\cup C\cup I=M$;
\item
$U\cap C=U\cap I=C\cap I=\emptyset$;
\item
besides, one or two subsets of $U$, $C$, and $I$ may be empty.
\end{enumerate}
\end{definition}

\begin{definition}
For the partition above, we can associate $U$, $C$ and $I$ with the subsets of the set of $Q$-terms $\left\{f_i\right\}_{i\in M}$ such that:
\begin{equation}\label{qdd}
\begin{aligned}
(1)\: &\text{for every $i\in U$ we have a $Q$-term $f_i$ that is an unconditional, and }\\
        &f_i=w^i;\\
(2)\: &\text{for every $i\in C$ we have a $Q$-term $f_i$ that is a conditional, and }\\
        &f_i=\left\{\left(u_j^i,w_j^i\right)\right\}_{j\in\{1,\dots, l(i)\}},\text{ where } l(i)\text{ is either a constant or}\\
        &\text{a value of a computable function of }N\text{ if }N\neq\emptyset;\\
(3)\: &\text{for every $i\in I$ we have a $Q$-term $f_i$ that is a conditional infinite, and }\\
        &f_i=\left\{\left(u_j^i,w_j^i\right)\right\}_{j\in\{1,2,\dots\}}.
\end{aligned}
\end{equation}
\end{definition}

\begin{remark}
If the operations form a chain, then they can be performed in arbitrary order including a \emph{doubling scheme}.
For example, the doubling scheme for computing the chain
\[
a_1+a_2+a_3+a_4+a_5+a_6+a_7+a_8
\]
is the following.
First, we compute
\[
c_1=a_1+a_2, c_2=a_3+a_4, c_3=a_5+a_6, \text{ and } c_4=a_7+a_8
\]
simultaneously.
Then
\[
d_1=c_1+c_2 \text{ and } d_2=c_3+c_4
\]
simultaneously.
After that,
\[
e=d_1+d_2.
\]
\end{remark}

\begin{definition}
Now we describe the promised implementation of the algorithm $\mathcal{A}$ that are called the \emph{$Q$-effective implementation of the algorithm $\mathcal{A}$}.
\begin{description}
\item[First, $N = \emptyset$.]
Let us have an interpretation of the variables of $B$.

We compute the expressions
\begin{multline}\label{W}
W=\bigl\{w^i (i\in U); u_j^i,w_j^i (i\in C,j\in\{1,\dots,l(i)\});\\
 u_j^i,w_j^i (i\in I, j\in\{1,2,\dots\})\bigr\}
\end{multline}
simultaneously, in parallel.

We say that \emph{the operation is ready to perform} if we have already computed the values of all its operands.
When computing each of the expressions of $W$ (cf. (\ref{W})), we perform the operations as soon as they are ready to be executed.
If several operations of a chain are ready for execution, then their computations are performed according to the doubling scheme.

If for any $i\in C\cup I$ and $j\in\{1,2,\dots\}$ we have the expression $u_j^i$ with the value \textsf{false}, then the computation of the corresponding expression of $w_j^i$ is terminated.

If for any $i\in C\cup I$ and $j\in\{1,2,\dots\}$ the computation of some pair of expressions $(u_j^i,w_j^i)$ has a consequence that the value of one of two expressions is not defined, then the computation of the other expression is terminated.

If for any $i\in C\cup I$ the computation of a certain pair of expressions $(u_{j_0}^i,w_{j_0}^i)$ leads to the determination of their values and $u_{j_0}^i$ is \textsf{true}, 
then the computation of expressions $u_j^i,w_j^i$ is terminated for any $j\neq j_0$.

Computation of identical expressions $W$ and their identical subexpressions may not be duplicated.
\item[Now, $N\neq\emptyset$.]
Let us have an interpretation of the variables of $B$ and specified $\bar{N}\in\{\bar{N}\}$.

We get the set of expressions
\begin{multline}\label{WN}
W(\bar{N})=\bigl\{w^i(\bar{N}) (i\in U); u_j^i(\bar{N}),w_j^i(\bar{N}) (i\in C, j\in\{1,\dots ,l(i)\});\\
u_j^i(\bar{N}),w_j^i(\bar{N}) (i\in I, j\in\{1,2,\dots\})\bigr\}.
\end{multline}
The expressions from $W(\bar{N})$ can be computed by analogy with computations of the expressions from $W$ (cf. (\ref{W})).
\end{description}
\end{definition}

\begin{remark}\label{rem:Qef}
The definition of the $Q$-effective implementation shows it is \emph{the most parallel implementation of the algorithm}.
In other words, the $Q$-effective implementation uses the parallelism resource of the algorithm completely, cf \ref{ss:pr}.
\end{remark}

\subsection{The concept of a realizable implementation of an algorithm}

It is very important how we can realize an implementation of an algorithm.

\begin{definition}
Let the algorithm $\mathcal{A}$ be represented in the form of a $Q$-deter\-minant $y_i=f_i$ for all $i\in M$ (cf. (\ref{fqd})).
An implementation of the algorithm  $\mathcal{A}$ is called \emph{realizable} if it is such that a finite number of operations must be performed simultaneously.
\end{definition}

There are algorithms such that the $Q$-effective implementation is not realizable.
\begin{example}
Compute the sum of a series
\[
S=\sum_{k=1}^\infty(-1)^k\frac{1}{k}
\]
with a given accuracy $\epsilon<1$.

The $Q$-determinant of the algorithm for computing $S$ consists of one conditional infinite $Q$-term (cf. (\ref{qd})).

Namely, the representation of the algorithm for computing $S$ in the form of a $Q$-determinant is written as
\begin{multline*}
S=\bigl\{\left(\tfrac{1}{2}<\epsilon,-1\right),\left(\tfrac{1}{3}<\epsilon,-1+\tfrac{1}{2}\right),\dots,\\
\left(\tfrac{1}{k}<\epsilon,-1+\tfrac{1}{2}-\dots+(-1)^{k-1}\tfrac{1}{k-1}\right),\dots\bigr\}.
\end{multline*}
As the countable set of division operations is ready to be performed simultaneously, then the $Q$-effective implementation isn't realizable.
\end{example}

To perform the $Q$-effective implementation, we specify some conditions on the $Q$-determinant.
\begin{theorem}\label{th:real}
Let the algorithm $\mathcal{A}$ be represented in the form of a $Q$-determinant $y_i=f_i$ for all $i\in M$.
Then the $Q$-effective implementation of the algorithm $\mathcal{A}$ is realizable, if one of the following three conditions is satisfied.
\begin{enumerate}
\item
 We have $I=\emptyset$.
\item
 We have $I\neq\emptyset$, $N = \emptyset$, and for every $r\in\{1,2,\dots\}$ the set of operations of the nesting level $r$ for the expressions $u_j^i$, $w_j^i$  is finite for all $i\in I$, and $j\in\{1,2,\dots\}$.
\item
 We have $I\neq\emptyset$, $N\neq\emptyset$, and for every $r\in\{1,2,\dots\}$ the set of operations of the nesting level $r$ for the expressions $u_j^i(\bar{N})$, $w_j^i(\bar{N})$  is finite for all $i\in I$, $j\in\{1,2,\dots\}$,
 and $\bar{N}\in\{\bar{N}\}$.
\end{enumerate}
\end{theorem}
\begin{proof}
Consider condition 1.
Let $I=N=\emptyset$.
Since $I=\emptyset$, executing the $Q$-effective implementation requires to compute a finite set of expressions
\[
W=\left\{w^i (i\in U); u_j^i,w_j^i (i\in C,j\in\{1,\dots ,l(i)\})\right\}.
\]
So, it is necessary to perform a finite number of operations simultaneously.

Let $N\neq\emptyset$.
In this case, when executing the $Q$-effective implementation for any $\bar{N}\in \{\bar{N}\}$, it is necessary to compute a finite set of expressions
\[
W(\bar{N})=\left\{w^i(\bar{N}) (i\in U); u_j^i(\bar{N}),w_j^i(\bar{N}) (i\in C,j\in\{1,\dots ,l(i)\})\right\}.
\]
Again, it is necessary to perform a finite number of operations simultaneously.

Now we consider condition 2.
It follows that for expressions $W$ (cf. (\ref{W})) the set of operations of the nesting level $r$ is finite for all $r\in\{1,2,\dots\}$.
Therefore, it is necessary to perform a finite number of operations under the $Q$-effective implementation simultaneously.

So, the $Q$-effective implementation of the algorithm $\mathcal{A}$ is realizable.

Consideration of condition 3 is similar to condition 2.
\end{proof}

\begin{remark}
We examined a considerable number of numerical algorithms and came to the conclusion that almost all of them have the realizable $Q$-effective implementation.
\end{remark}

\subsection{The concept of the parallelism resource of an algorithm}\label{ss:pr}

\begin{requirements}
We hold the following conditions and notations.
\begin{enumerate}
\item
The algorithm $\mathcal{A}$ is represented in the form of a $Q$-determinant $y_i=f_i$ for all $i\in M$ (cf. (\ref{fqd})).
\item
The values of the $Q$-terms $f_i$ for all $i\in M$ is determined under any interpretation of the variables of $B$ and any $\bar{N}\in \{\bar{N}\}$ if $N\neq\emptyset$.
\item
The $Q$-effective implementation of the algorithm $\mathcal{A}$ is realizable.
\item
Let $N=\emptyset$ and $I\neq\emptyset$.
Then under a given interpretation of the variables of $B$ for any $i\in I$ there is a pair of expressions $u_{j_i}^i$, $w_{j_i}^i$ such that the value of $u_{j_i}^i$ is equal to \textsf{true}, and the value of $w_{j_i}^i$ is defined.

Note that $j_i$ depends on the interpretation of the variables of $B$.
We introduce the notation
\[
\widetilde{W}=\left\{w^i (i\in U); u_j^i,w_j^i (i\in C,j\in\{1,\dots,l(i)\});u_{j_i}^i,w_{j_i}^i (i\in I)\right\}.
\]
\item
Let $N\neq\emptyset$ and $I\neq\emptyset$.
Then under a given interpretation of the variables of $B$, for any $\bar{N}\in \{\bar{N}\}$, and $i\in I$ there is a pair of expressions
$u_{j_i}^i(\bar{N})$, $w_{j_i}^i(\bar{N})$ such that the value of $u_{j_i}^i(\bar{N})$ is equal to \textsf{true}, and the value of $w_{j_i}^i(\bar{N})$ is defined.

Note that $j_i$ depends on $\bar{N}$ and the interpretation of the variables of $B$.
We introduce the notation
\begin{multline*}
\widetilde{W}(\bar{N})=\bigl\{w^i(\bar{N}) (i\in U); u_j^i(\bar{N}),w_j^i(\bar{N}) (i\in C,j\in\{1,\dots,l(i)\});\\
 u_{j_i}^i(\bar{N}),w_{j_i}^i(\bar{N}) (i\in I)\bigr\}.
\end{multline*}
\end{enumerate}
\end{requirements}

\begin{definition}
We define the characteristics of the parallelism resource of the algorithm $\mathcal{A}$:
\[
D_\mathcal{A}\text{\emph{ is the algorithm height and }}P_\mathcal{A}\text{\emph{ is the algorithm width}}.
\]

If $N=\emptyset$, then
\begin{flalign}
D_\mathcal{A}&=\begin{cases}\label{D}
\max\limits_{w\in W}T^w &\text{if } I=\emptyset, \\
\max\limits_{w\in\widetilde{W}}T^w& \text{if } I\neq\emptyset;
\end{cases}
\intertext{if $O_r^{w}$ is the number of operations of the nesting level $r$ of the expression $w$, then}
P_\mathcal{A}&=\max\limits_{1\leq r\leq D_\mathcal{A}}\sum_{w\in W}O_r^w.\label{P}
\end{flalign}

If $N\neq\emptyset$, then
\begin{flalign}
D_\mathcal{A}(\bar{N})&=\begin{cases}\label{DN}
\max\limits_{w(\bar{N})\in W(\bar{N})}T^{w(\bar{N})} \text{ if } I=\emptyset, \\
\max\limits_{w(\bar{N})\in \widetilde{W}(\bar{N})}T^{w(\bar{N})} \text{ if } I\neq\emptyset;
\end{cases}
\intertext{if $O_r^{w(\bar{N})}$ is the number of operations of the nesting level $r$ of the expression $w(\bar{N})$, then}
P_\mathcal{A}(\bar{N})&{=}\max\limits_{1\leq r\leq D_\mathcal{A}({\bar{N})}}\sum_{w(\bar{N})\in W(\bar{N})}O_r^{w(\bar{N})}.\label{PN}
\end{flalign}

In these formulas, $W$ and $W(\bar{N})$ have the same meaning as in formulas (\ref{W}) and (\ref{WN}), respectively.
\end{definition}

\begin{remark}\label{rem:DP}
We would like to note some important features of $D_\mathcal{A}$ and $P_\mathcal{A}$.
\begin{enumerate}
\item
$D_\mathcal{A}$ and $P_\mathcal{A}$ depend on $N$ if $N\neq\emptyset$.
\item
$D_\mathcal{A}$ and $P_\mathcal{A}$ don't depend on the interpretation of the variables of $B$ if $I=\emptyset$.
\item
$D_\mathcal{A}$ and $P_\mathcal{A}$ depend on the interpretation of the variables of $B$ if $I\neq\emptyset$.
\item
The values of $D_\mathcal{A}$ and $P_\mathcal{A}$ estimate the parallelism resource of the algorithm $\mathcal{A}$.
More exactly, $D_\mathcal{A}$ characterizes \emph{the execution time of the $Q$-effective implementation of the algorithm},
and $P_\mathcal{A}$ characterizes \emph{the number of processors required to execute the $Q$-effective implementation}.
\end{enumerate}
\end{remark}

\section{The $Q$-determinants, realizabilities, and the parallelism resources of some numerical algorithms}

In this section, we consider three algorithms: the scalar product, the Gaussian elimination, and solving a system of grid equations by the Jacobi method.
We would like to point out the reasons for choosing these algorithms.
\begin{description}
\item[First.]
Although the computation of the scalar product of vectors is very simple, this computation is very common and is an inseparable part of many algorithms.
So, it seems to us that the consideration is very useful and important.

The $Q$-determinant  consists of one unconditional $Q$-term.
\item[Second.]
It is quite clear that the Gaussian elimination in various forms is one of the bases of numerical mathematics, but even of all mathematics.

In this case, the $Q$-determinant consists of $n$ conditional $Q$-terms of length $n!$, where $n$ is an integer.
\item[Third.]
The solving a system of grid equations by the Jacobi method is a well-known iterative method.
We need to consider the Jacobi method, because the number of iterations is a very important characteristic of many numerical algorithms.

Now the $Q$-determinant consists of a finite number of conditional infinite $Q$-terms.
\end{description}
Thus, in our opinion, we will consider very important methods with different structures of $Q$-determinants.

\subsection{The scalar product of vectors}

\subsubsection{The $Q$-determinant of the scalar product}

Consider the algorithm $\mathcal{S}$ for computing the scalar product of vectors.
\[
\vec{a}^1=(a_1^1,\dots,a_n^1) \text{ and } \vec{a}^2=(a_1^2,\dots,a_n^2)
\]
\begin{flalign}\label{eq:sp}
(\vec{a}^1,\vec{a}^2)&=\sum_{i=1}^na_i^1 a_i^2.
\end{flalign}	
\begin{qdet}
In this case $N=\{n\}$, $B=\{a_i^1,a_i^2 \mid i\in\{1,\dots n\}\}$, and $\vec{y}=\{(\vec{a}^1,\vec{a}^2)\}$.
Now we note that the equation {\rm(\ref{eq:sp})} is the representation of the algorithm $\mathcal{S}$ in the form of a $Q$-determinant.
So, the $Q$-determinant consists of one unconditional $Q$-term.
\end{qdet}

\subsubsection{Realizability and the parallelism resource of the scalar product}

Let $\log x$ be the binary logarithm of the number $x$, and $\lceil x\rceil$ be the ceiling of $x$, i.e., the least integer greater than or equal to the number $x$.
\begin{proposition}\label{p:s}
The $Q$-effective implementation of the algorithm $\mathcal{S}$ is realizable.

For the algorithm $\mathcal{S}$ we have
\[
D_\mathcal{S}=\left\lceil\log n\right\rceil+1\text{ and }P_\mathcal{S}=n.
 \]
\end{proposition}
\begin{proof}
Indeed, since $I=\emptyset$, then the $Q$-effective implementation of the algorithm $\mathcal{S}$ is realizable by Theorem \ref{th:real}.

In this case, we must use the doubling scheme.
Therefore, the definitions of height and width for the algorithm $\mathcal{S}$ are completely obvious.
\end{proof}

\subsection{The Gauss--Jordan method for solving a system of linear equations}

It is well known the Gauss--Jordan method (the Gaussian elimination) is universal.

For simplicity, suppose that a $n\times n$ matrix $A$ is invertible (has a nonzero determinant).

Let $\vec{x}=(x_1,\dots ,x_n)^T$, $\vec{b}=(a_{1,n+1},\dots ,a_{n,n+1})^T$ be column vectors, and $\bar{A}=[a_{ij}] $ be an augmented matrix of the system.
Therefore,
\[
A\vec{x}=\vec{b}.
\]
In this case
\begin{gather*}
N=\{n\},\ B=\{a_{ij} \mid i\in\{1,\dots n\}, j\in\{1,\dots n+1\}\},\text{ and}\\
\vec{y}=\{x_i \mid i\in\{1,\dots n\}\}.
\end{gather*}

\subsubsection{The $Q$-determinant of the Gauss--Jordan method}\label{sss:G}

There are many variants of algorithms for implementation of  the Gauss--Jordan method.

Consider one of them that denote by $\mathcal{G}$.
This algorithm has $n$ steps.
\begin{description}
\item[Step $1$.]
We must select the leading element.
If $a_{11}\neq 0$, then $a_{11}$ is the leading element and $j_1=1$.
Otherwise, if $a_{1j}=0$ for $j<j_1\leq n$ and $a_{1j_1}\neq 0$, then $a_{1j_1}$ is the leading element.
So, the first non-zero element in the first row of the matrix $A$ is the leading element.

Then we get the updated augmented matrix $\bar{A}^{j_1}=[a_{ij}^{j_1}]$ by the rule
\begin{align*}
a_{1j}^{j_1}&=\frac{a_{1j}}{a_{1j_1}},\\
a_{ij}^{j_1}&=a_{ij}-\frac{a_{1j}}{a_{1j_1}}a_{ij_1}
\end{align*}
for every $i\in\{2,\dots,n\}$ and $j\in\{1,\dots,n+1\}$.

The nesting level of the right-hand sides of the equations is at most $3$.

More exactly, for every $n\in\{2,3,\dots\}$ and $j\in\{1,\dots ,n+1\}$ the expression
\[
\frac{a_{1j}}{a_{1j_1}}
\]
has the nesting level $1$, and for every $i\in\{2,\dots ,n\}$ the expression
\[
a_{ij}-\frac{a_{1j}}{a_{1j_1}}a_{ij_1}
\]
has the nesting level $3$.

\item[Step $k\in\{2,\dots,n\}$.]
After step $k-1$ we get an augmented matrix $\bar{A}^{j_1\dots j_{k-1}}$.
We must select the leading element.

If $a_{k1}^{j_1\dots j_{k-1}}\neq 0$, then $a_{k1}^{j_1\dots j_{k-1}}$ is the leading element and $j_k=1$.
Otherwise, if $a_{kj}^{j_1\dots j_{k-1}}=0$ for $j<j_k\leq n$ and $a_{kj_k}^{j_1\dots j_{k-1}}\neq 0$, then $a_{kj_k}^{j_1\dots j_{k-1}}$ is the leading element.
So, the first nonzero element in $k$th row of the matrix $\bar{A}^{j_1\dots j_{k-1}}$ is the leading element.

Now we obtain the next augmented matrix
\[
\bar A^{j_1\dots j_k}=\left[a_{ij}^{j_1\dots j_k}\right]_{\begin{subarray}{l}i\in\{1,\dots ,n\},\\j\in\{1,\dots ,n+1\}\end{subarray}}
\]
by rule
\begin{align*}
a_{kj}^{j_1\dots j_k}&=\frac{a_{kj}^{j_1\dots j_{k-1}}}{a_{kj_k}^{j_1\dots j_{k-1}}},\\
a_{ij}^{j_1\dots j_k}&=a_{ij}^{j_1\dots j_{k-1}}-\frac{a_{kj}^{j_1\dots j_{k-1}}}{a_{kj_k}^{j_1\dots j_{k-1}}}a_{ij_k}^{j_1\dots j_{k-1}}
\end{align*}
for every $i\in\{1,\dots,n\}$, $i\neq k$, and $j\in\{1,\dots,n+1\}$.

By induction, the nesting level of the right-hand sides of the equations is at most $3k$.

More exactly, for every $n\in\{2,3,\dots\}$ and $j\in\{1,\dots ,n+1\}$ the expression
\[
\frac{a_{kj}^{j_1\dots j_{k-1}}}{a_{kj_k}^{j_1\dots j_{k-1}}}
\]
has the nesting level $3(k-1)+1=3k-2$, and for every  $i\in\{1,\dots,n\}$, $i\neq k$ the expression
\[
a_{ij}^{j_1\dots j_{k-1}}-\frac{a_{kj}^{j_1\dots j_{k-1}}}{a_{kj_k}^{j_1\dots j_{k-1}}}a_{ij_k}^{j_1\dots j_{k-1}}
\]
has the nesting level $3k$.
\end{description}

As a result, we get a system of equations $A^{j_1\dots j_n}\vec{x}=\vec b^{j_1\dots j_n}$ after step $n$, where
\begin{align*}
A^{j_1\dots j_n}&=\left[a_{ij}^{j_1\dots j_n}\right]_{\begin{subarray}{l}i\in\{1,\dots ,n\},\\j\in\{1,\dots ,n+1\}\end{subarray}},\\
\vec{b}^{j_1\dots j_n}&=(a_{1,n+1}^{j_1\dots j_n},\dots ,a_{n,n+1}^{j_1\dots j_n})^T.
\end{align*}

Moreover, for every $i\in\{1,\dots ,n\}$
\[
a_{ij_i}^{j_1\dots j_n}=1
\]
and for every $i\in\{1,\dots ,n\}$, $j\in\{1,\dots ,n\}$, and $j\neq j_i$
\[
a_{ij}^{j_1\dots j_n}=0.
\]
Hence, for every $i\in\{1,\dots ,n\}$
\[
x_{j_i}=a_{i,n+1}^{j_1\dots j_n}
\]
that is the solution of our system.

We also note that for every $n\in\{2,3,\dots\}$
\begin{equation}\label{f:nl}
T^w=
\begin{cases}
3n&\text{for }w=a_{i,n+1}^{j_1\dots j_n},\text{ if }i\in\{1,\dots ,n-1\};\\
3n-2&\text{for }w=a_{n,n+1}^{j_1\dots j_n}.
\end{cases}
\end{equation}

The $n$-tuple $(j_1,\dots ,j_n)$ determines the choice of leading elements.
This is the permutation of elements of the set $\{1,\dots, n\}$.
The number of such permutations is $n!$.
We number these permutations by positive integers from $1$ to $n!$.
So, each permutation has own serial number.

Denote by
\begin{align*}
L_{j_1}&=
\begin{cases}
\text{\textsf{true}}&\text{if }j_1=1,\\
\displaystyle{\bigwedge_{j=1}^{j_1-1}(a_{1j}=0)}&\text{ if }j_1\neq 1
\end{cases}\\
\intertext{and for every $l\in\{2,\dots ,n\}$}
L_{j_l}&=
\begin{cases}
\text{\textsf{true}}&\text{if }j_l=1,\\
\displaystyle{\bigwedge_{j=1}^{j_l-1}(a_{lj}^{j_1\dots j_{l-1}}=0)}& \text{if }j_l\neq 1.
\end{cases}
\end{align*}

We have the parameter set $N=\{n\}$ and $M=\{1,\dots,n\}$ as in (\ref{ap}) and (\ref{qd}).
Further, let $p$ be a serial number of permutation $(j_1,\dots,j_n)$.
Then
\begin{equation}
w_p^{j_l}=a_{l,n+1}^{j_1\dots j_n}
\label{f:w}
\end{equation}
for every $l\in\{1,\dots ,n\}$ and
\begin{equation}
u_p=L_{j_1}\wedge (a_{1j_1}\neq 0)\wedge\left(\bigwedge_{l=2}^n\left(L_{j_l}\wedge\left(a_{l j_l}^{j_1\dots j_{l-1}}\neq 0\right)\right)\right)\label{uiGJ}
\end{equation}
are unconditional $Q$-terms.
\begin{qdet}\label{qd:G}
Finally, we have the set
\begin{equation}
x_j=\left\{(u_1,w_1^j),\dots ,(u_{n!},w_{n!}^j)\right\}_{j\in\{1,\dots ,n\}}
\end{equation}
is the representation of the algorithm $\mathcal{G}$ in the form of a $Q$-determinant.

Hence, the $Q$-determinant consists of $n$ conditional $Q$-terms of length $n!$.
\end{qdet}

\subsubsection{Realizability and the parallelism resource of the Gauss--Jordan method}

\begin{lemma}\label{l:realG}
The $Q$-effective implementation of the algorithm $\mathcal{G}$ is realizable.
\end{lemma}
\begin{proof}
Since $I=\emptyset$, then the $Q$-effective implementation of the algorithm $\mathcal{S}$ is realizable by Theorem \ref{th:real}.
\end{proof}

We need two auxiliary Lemmas.
First of all, we define an additional function $m_n$ as follows:
\begin{align*}
m_n(0)&=n,
\intertext{for every  $k\in\{0,1,\dots,n-2\}$}
m_n(k+1)&=\left\lceil\frac{m_n(k)}{8}\right\rceil+n-(k+1).
\end{align*}

\begin{lemma}\label{l:g1}
For every $k\in\{0,1,\dots,n-1\}$
\[
m_n(k)\geq m_{n-1}(k)\geq m_n(k)-2.
\]
\end{lemma}
\begin{proof}
We prove by induction on $k$.

For $k=0$,
\[
m_n(0)=n>m_{n-1}(0)=n-1>m_n(0)-2=n-2,
\]
that we need.

By the inductive assumption,
\[
m_n(k)\geq m_{n-1}(k)\geq m_n(k)-2.
\]
From that,
\[
\left\lceil\frac{m_n(k)}{8}\right\rceil\geq\left\lceil\frac{m_{n-1}(k)}{8}\right\rceil\geq\left\lceil\frac{m_n(k)-2}{8}\right\rceil.
\]
Then
\[
\left\lceil\frac{m_n(k)}{8}\right\rceil+n-(k+1)\geq\left\lceil\frac{m_{n-1}(k)}{8}\right\rceil+n-(k+1)\geq\left\lceil\frac{m_n(k)-2}{8}\right\rceil+n-(k+1).
\]
Hence,
\begin{gather*}
\left\lceil\frac{m_n(k)}{8}\right\rceil+n-(k+1)=m_n(k+1),\\
\left\lceil\frac{m_{n-1}(k)}{8}\right\rceil+n-(k+1)=m_{n-1}(k+1)+1,\\
\left\lceil\frac{m_n(k)-2}{8}\right\rceil+n-(k+1)\geq\left\lceil\frac{m_n(k)}{8}\right\rceil-1+n-(k+1)=m_n(k+1)-1.
\end{gather*}
From that, it follows
\[
m_n(k+1)\geq m_{n-1}(k+1)\geq m_n(k+1)-2,
\]
that we need.
\end{proof}

\begin{lemma}
Let $p_0$ be a serial number of the permutation $(n,n-1,\dots ,1)$.
For every $n\geq 2$ the nesting level
\[
T^{u_{p_0}}=3n-1.
\]
Moreover, for every $p\in\{1,\dots,n!\}$ the nesting level
\[
T^{u_p}\leq3n-1.
\]
\end{lemma}
\begin{proof}
By formula (\ref{uiGJ}), for every $p\in\{1,\dots,n!\}$ the nesting level $T^{u_p}$ of the $Q$-term $u_p$ does not exceed the nesting level $T^{u_{p_0}}$ of the $Q$-term $u_{p_0}$.
So we should only determine $T^{u_{p_0}}$.

The nesting level of the $Q$-term $u_{p_0}$ is equal to
\[
T^{u_{p_0}}=3n-2+\left\lceil\log m_n(n-1)\right\rceil.
\]
Hence, it is sufficient to prove
\[
m_n(n-1)=2.
\]

We prove it by induction on $n$.
For $n=2$ we have
\[
m_2(1)=\left\lceil\frac{2}{8}\right\rceil+1=2.
\]

By the inductive assumption, there is $m_{n-1}(n-2)=2$.
By Lemma \ref{l:g1},
\[
m_n(n-2)\geq m_{n-1}(n-2)\geq m_n(n-2)-2.
\]
So,
\[
m_n(n-2)\geq 2\geq m_n(n-2)-2\text{ or }2\leq m_n(n-2)\leq 4.
\]
Then
\[
\left\lceil\frac{2}{8}\right\rceil+1\leq\left\lceil\frac{m_n(n-2)}{8}\right\rceil+1\leq\left\lceil\frac{4}{8}\right\rceil+1.
\]
Since
\[
\left\lceil\frac{m_n(n-2)}{8}\right\rceil+1=m_n(n-1),
\]
then
\[
m_n(n-1)=2\text{ and }T^{u_{p_0}}=3n-1.
\]
\end{proof}

\begin{proposition}\label{p:G}
The $Q$-effective implementation of the algorithm $\mathcal{G}$ is realizable.

For the algorithm $\mathcal{G}$
\[
D_\mathcal{G}=3n\text{ and }P_\mathcal{G}\geq\frac{3}{2}(n+1)!
 \]
for $n\geq2$.
\end{proposition}
\begin{proof}
We obtain realizability by Lemma \ref{l:realG}.

From (\ref{f:nl}) and (\ref{f:w}) the nesting level
\[
T^{w_p^j}\leq3n
\]
for every $p\in\{1,\dots ,n!\}$ and $j\in\{1,\dots ,n\}$, and for some $p$ and $j$
\[
T^{w_p^j}=3n.
\]
So, the height of the considered algorithm $\mathcal{G}$ is
\[
D_\mathcal{G}=3n.
\]

Finally, we estimate the width of the algorithm $\mathcal{G}$.
The number of operations of the first nesting level of the set of expressions $W(\bar{N})$ (cf. (\ref{WN})) of the algorithm $\mathcal{G}$ is
\[
\sum_{k=1}^n{((k+n+1)(n-1)!)}=\frac{3}{2}(n+1)n!=\frac{3}{2}(n+1)!.
\]
Therefore, the width of the considered algorithm $\mathcal{G}$ is
\[
P_\mathcal{G}\geq \frac{3}{2}(n+1)!.
\]
\end{proof}

\begin{remark}
If the algorithm $\mathcal{B}$ implements the Gauss--Jordan method, and the leading elements satisfy the conditions
\begin{align*}
\left|a_{1j_1}\right|&=\max_{1\leq j\leq n}\left|a_{1j}\right|
\intertext{and for every $l\in\{2,\dots,n\}$}
\left|a_{l j_l}^{j_1\dots j_{l-1}}\right|&=\max_{\substack{1\leq j\leq n\\ j\notin\{j_1,\dots ,j_{l -1}\}}}\left|a_{l j}^{j_1\dots j_{l -1}}\right|,
\end{align*}
then $D_\mathcal{B}=3n$ also.
The proof is similar to the proof for the algorithm $\mathcal{G}$.
\end{remark}

\subsection{Solving a system of grid equations by\\ the Jacobi method}\label{ss:gJ}

In this subsection we consider some method for solving a system of grid equations.
This method has two sources: the classical iterative Jacobi method for solving a system of linear equations (see also \ref{ss:JM} on page ~\pageref{ss:JM}) and a system of grid equations for numerical solving the Poisson equation, for example,
\cite[Chapter 5, \S~1] {al:SamGul}.

Therefore, this method may be called as the Jacobi method for solving a system of grid equations.
We would also like to note that this method is a model for us, and many similar considerations can be used in other cases.
As we will not consider the practical application of this method, then we will not pay attention to its convergence and other similar topics.

\subsubsection{The $Q$-determinant of the Jacobi method}\label{sss:qdJ}

Suppose we have a five-point system of linear equations
\begin{multline*}
u_{kj}=\frac{f_{kj}+a_{kj}u_{k-1,j}+b_{kj}u_{k,j-1}+c_{kj}u_{k+1,j}+d_{kj}u_{k,j+1}}{e_{kj}}\\ 
(k\in\{1,\dots, K\}, j\in\{1,\dots,J\}),
\end{multline*}
where $u_{kj}$ are the values of the grid function, $a_{kj},b_{kj},c_{kj},d_{kj},e_{kj},f_{kj}$ are constants for every $k\in\{1,\dots, K\}$ and $j\in\{1,\dots,J\}$.
As usual, we set
\[
u_{0j}=u_{k0}=u_{K+1,j}=u_{k,J+1}=0.
\]

Let the algorithm $\mathcal{J}$ implement the Jacobi method for solving this system.

For every $n\in\{1,2,\dots\}$
\[
u_{kj}^n=\frac{f_{kj}+a_{kj}u_{k-1,j}^{n-1}+b_{kj}u_{k,j-1}^{n-1}+c_{kj}u_{k+1,j}^{n-1}+d_{kj}u_{k,j+1}^{n-1}}{e_{kj}},
\]
where $u_{kj}^0$ are the initial approximate values of the grid function and $u_{kj}^n$ are the approximate values of the grid function obtained at the iteration step $n$ for all $k\in\{1,\dots, K\}$ and $j\in\{1,\dots,J\}$.

Let
\[
V_n=\left\{u_{kj}^n\mid k\in\{1,\dots, K\},\,j\in\{1,\dots,J\}\right\}
\]
be the set of values of the grid function at the $n$th iteration step.

We express $V_{n-1}$ in terms of $V_{n-2}$.
Then we can express $V_n$ in terms of $V_{n-2}$.
Moreover, for all $k\in\{1,\dots, K\}$ and $j\in\{1,\dots,J\}$  we can represent $u_{kj}^n$ by the elements from
\[
\left\{u_{sl}^{n-2}\mid (s,l)\in I_2(k,j)\right\},
\]
where for all $k\in\{1,\dots, K\}$ and $j\in\{1,\dots,J\}$
\[
I_2(k,j)=\left\{(s,l)\mid |s-k|+|l-j|\leq2\right\}.
\]
Continuing this process, for every $k\in\{1,\dots, K\}$ and $j\in\{1,\dots,J\}$  we can represent $u_{kj}^n$ by the elements from
\[
\left\{u_{sl}^0\mid (s,l)\in I_n(k,j)\right\},
\]
where for all $k\in\{1,\dots, K\}$ and $j\in\{1,\dots,J\}$
\[
I_n(k,j)=\left\{(s,l)\mid  |s-k|+|l-j|\leq n\right\}.
\]
Let $\epsilon$ be a given accuracy.
Suppose the iterative process is finished if
\[
v^n=\bigwedge_{\substack{k\in\{1,\dots, K\},\\ j\in\{1,\dots,J\}}}(|u_{kj}^n-u_{kj}^{n-1}|<\epsilon)
\]
takes the value \textsf{true}.
We can represent the expression $v^n$ through the set $V_0$.
Further, note that $v^n$ and $u_{kj}^n$ are the $Q$-terms for every $k\in\{1,\dots, K\}$, $j\in\{1,\dots,J\}$, and $n\in\{1,2,\dots,\}$.
Also
\[
N=\{K,J\}\text{ and }B=\left\{u_{sl}^0\mid s\in\{1,\dots,K\},\, l\in\{1,\dots,J\}\right\}.
\]
\begin{qdet}
The $Q$-determinant of the algorithm $\mathcal{J}$ consists of $KJ$ conditional infinite $Q$-terms, and the representation in the form of a $Q$-determinant is
\begin{equation}
u_{kj}=\left\{(v^1,u_{kj}^1),(v^2,u_{kj}^2),\dots ,(v^n,u_{kj}^n),\dots\right\},
\end{equation}
where $k\in\{1,\dots,K\}$, $j\in\{1,\dots,J\}$, and $n\in\{1,2,\dots\}$.
\end{qdet}

\subsubsection{Realizability and the height of the Jacobi method}

\begin{lemma}\label{l:realJ}
The $Q$-effective implementation of the algorithm $\mathcal{J}$ is realizable.

If the integer $n_0$ is such that the value of $v^{n_0}$ is \textsf{true}, then the height of the algorithm $J$ is
\[
D_\mathcal{J}=5n_0+3+\lceil\log{KJ}\rceil.
\]
\end{lemma}
\begin{proof}
The $Q$-effective implementation of the algorithm $\mathcal{J}$ is realizable by Theorem \ref{th:real} (Statement (3)) because we have a finite set of operations for every nesting level.

If for some $n_0$ the value of $v^{n_0}$ is \textsf{true}, then performing the $Q$-effective implementation should be completed, and the value of $u_{kj}^{n_0}$ should be taken as a solution $u_{kj}$.

Since $N\neq\emptyset$, the height of the algorithm $\mathcal{J}$ should be computed by the formula (\ref{DN}) for the case $I\neq\emptyset$.
In this case,
\[
\widetilde{W}(\bar{N})=\left\{v^{n_0}; V_{n_0}=\left\{u_{kj}^{n_0}\mid k\in\{1,\dots, K\},\,j\in\{1,\dots,J\}\right\}\right\}.
\]
Taking into account the formula for calculating $v^{n_0}$, we have
\[
T^{v^{n_0}}=5n_0+3+\lceil\log{KJ}\rceil.
\]
Thus, the maximum nesting level of the expressions of the set $\widetilde{W}(\bar{N})$ is equal to
\[
5n_0+3+\lceil\log{KJ}\rceil.
\]
Therefore, the height of the algorithm $\mathcal{J}$ is
\[
D_\mathcal{J}=5n_0+3+\lceil\log{KJ}\rceil.
\]
\end{proof}

\begin{remark}
It seems to us that the determination of the integer $n_0$ from Lemma \ref{l:realJ} is a very difficult problem in the general case.
It is possible that this is an unsolvable problem because it presupposes too much input data: the sets $\left\{a_{kj}\right\}_{k,j}$, $\left\{b_{kj}\right\}_{k,j}$, $\left\{c_{kj}\right\}_{k,j}$, $\left\{d_{kj}\right\}_{k,j}$, $\left\{e_{kj}\right\}_{k,j}$, and $\left\{f_{kj}\right\}_{k,j}$.
However, it is worth noting that the integer $n_0$ from Lemma \ref{l:realJ} can be found in \cite[Chapter 5, \S~1, p.~382]{al:SamGul} for the numerical solution of the Poisson equation which is an important case.
\end{remark}

\subsubsection{Performing the $Q$-effective implementation of the algorithm $\mathcal{J}$}

We describe the process of performing the $Q$-effective implementation of the algorithm $\mathcal{J}$, which is divided into stages.
\begin{description}
\item[Stage $1$.]
As in \ref{sss:qdJ}, performing the $Q$-effective implementation begins with the computation of a set of $Q$-terms
\[
V_1=\left\{u_{kj}^1\mid k\in\{1,\dots, K\},\,j\in\{1,\dots,J\}\right\}
\]
using $V_0=\left\{u_{kj}^0\mid k\in\{1,\dots, K\},\,j\in\{1,\dots,J\}\right\}$.
For that, the number of operations is $4KJ$ at the first nesting level, at the second nesting level is $2KJ$, at each of the nesting levels 3, 4, 5 is $KJ$.
\item[Stage $2$.]
Now we can compute $Q$-terms
\[
v^1\text{ and }V_2=\left\{u_{kj}^2\mid k\in\{1,\dots, K\},\,j\in\{1,\dots,J\}\right\}.
\]
Their computations must be performed simultaneously.

The number of operations for the computation of $v^1$ is $KJ$ at each of the first three levels of nesting.
After that, we have a chain of conjunctions of length $KJ$, computed by the doubling scheme.
To do this, we use two levels of nesting.

The computation of $V_2$ is similar to the computation of $V_1$.
\item[Stage $p\geq3$] can be described as follows.
Since we computed $V_{p-1}$, then we can continue the computation of $Q$-terms $v^1$,\dots,$v^{p-2}$.

Simultaneously, we can begin to compute $Q$-terms
\[
v^{p-1}\text{ and }V_p=\left\{u_{kj}^p\mid k\in\{1,\dots, K\},\,j\in\{1,\dots,J\}\right\}.
\]

Once again, the number of operations for the computation of $v^{p-1}$ is $KJ$ at each of the first three levels of nesting.
After that, we have a chain of conjunctions of length $KJ$, computed by the doubling scheme.
To do this, we use two levels of nesting.

The computation of $V_p$ is similar to the computation of $V_1$.
\end{description}

\begin{lemma}\label{l:st}
Every stage of the process for performing the $Q$-effective implementation of the algorithm $\mathcal{J}$ has five levels of nesting.
\end{lemma}
\begin{proof}
Based on the above, every stage $p$ has five levels of nesting, because we compute
\begin{align*}
V_p&=\left\{u_{kj}^p\mid k\in\{1,\dots, K\},\,j\in\{1,\dots,J\}\right\}
\intertext{on the base of}
V_{p-1}&=\left\{u_{kj}^{p-1}\mid k\in\{1,\dots, K\},\,j\in\{1,\dots,J\}\right\}.
\end{align*}
\end{proof}

 \subsubsection{The width of the Jacobi method}

\begin{remark}\label{r:p12}
It is very important to point out the following.
If we have only one stage, then the consideration doesn't make sense, since stage 2 is the beginning of the computation of $v^1$.

So we only consider $p\geq2$ stages.
\end{remark}

\begin{lemma}\label{l:nol+}
Suppose that we execute stage $p$ of the process for performing the $Q$-effective implementation of the algorithm $\mathcal{J}$.
\begin{enumerate}
\item
If $p\geq1$, then under the computation of 
\[
V_p=\left\{u_{kj}^p\mid k\in\{1,\dots, K\},\,j\in\{1,\dots,J\}\right\} 
\]
we must execute the following number of operations:
\begin{align*}
4KJ&\text{ for the first nesting level;}\\
2KJ&\text{ for the second nesting level;}\\
KJ&\text{ for the nesting levels $3$, $4$, and $5$.}
\end{align*}
\item
If $p\geq2$ and $KJ\geq4$, then under the computation of $v^{p-1}$ we must execute the following number of operations:
\begin{align*}
KJ&\text{ for the nesting levels $1$, $2$, and $3$;}\\
\lfloor KJ/2\rfloor&\text{ for the fourth nesting level;}\\
\lfloor KJ/4\rfloor&\text{ for the fifth nesting level.}
\end{align*}
\item
Assume that $p\geq3$ and $KJ\geq2^{5p-8}$.
Let $l\in\{1,2,3,4,5\}$ be a nesting level.
Then for every $t\geq2$ under the computation of $v^{p-t}$ we must execute the following number of operations:
\[
\left\lfloor\frac{KJ}{2^{2+l}\times32^{t-2}}\right\rfloor.
\]
\end{enumerate}
In particular, if $p\geq3$ and $KJ\geq2^{5p-8}$, then
\[
\left\lfloor\frac{KJ}{2^{2+l}\times32^{p-3}}\right\rfloor
\]
operations on the nesting level $l\in\{1,2,3,4,5\}$ to compute $v^1$.
\end{lemma}
\begin{proof}
The first and second stages ($p\in\{1,2\}$) differ from the others.
\begin{enumerate}
\item
In stage 1, we compute $V_1=\left\{u_{kj}^1\mid k\in\{1,\dots,K\},\,j\in\{1,\dots,J\}\right\}$.
For that, we must execute the following number of operations:
\begin{align*}
4KJ&\text{ for the first nesting level;}\\
2KJ&\text{ for the second nesting level;}\\
KJ&\text{ for the nesting levels $3$, $4$, and $5$.}
\end{align*}
\item
In stage 2, we compute $V_2=\left\{u_{kj}^2\mid k\in\{1,\dots,K\},\,j\in\{1,\dots,J\}\right\}$ and an expression $v^1$.
To compute $V_2$, we must have as many operations as in stage 1 for $V_1$.
To compute the expression $v^1$, we must execute the following numbers of operations:
\begin{align*}
KJ&\text{ for the nesting levels $1$, $2$, and $3$;}\\
\lfloor KJ/2\rfloor&\text{ for the fourth nesting level;}\\
\lfloor KJ/4\rfloor&\text{ for the fifth nesting level.}
\end{align*}
\item
Now we consider the case $p\geq3$.
First, note:
\begin{enumerate}
\item
it is clear that to compute
\[ 
V_p=\left\{u_{kj}^p\mid k\in\{1,\dots,K\},\,j\in\{1,\dots,J\}\right\} 
\]
we need as many operations as in stage 1 for $V_1$;
\item
also, the computation of the expression $v^{p-1}$ is needed the same number of operations as in stage 2 for $v_1$, because we are just beginning the computation of $v^{p-1}$.
\end{enumerate}
Therefore, we must consider the computation of $v^{p-t}$ for $t\geq2$.
We fix $t\geq2$.
We started computing $v^{p-t}$ in stage $p-t+1$.
By Statement 2, we must use the following number of operations:
\begin{align*}
KJ&\text{ for the nesting levels $1$, $2$, and $3$;}\\
\lfloor KJ/2\rfloor&\text{ for the fourth nesting level;}\\
\lfloor KJ/4\rfloor&\text{ for the fifth nesting levels}
\end{align*}
under the computation of $v^{p-t}$  in stage $p-t+1$.

We use the doubling scheme for the computation of $v^{p-t}$.
Therefore, in stage $p-t+2$ we must execute the following number of operations:
\begin{align*}
\lfloor KJ/8\rfloor&\text{ for the first nesting;}\\
\lfloor KJ/16\rfloor&\text{ for the second nesting;}\\
\lfloor KJ/32\rfloor&\text{ for the third nesting;}\\
\lfloor KJ/64\rfloor&\text{ for the fourth nesting level;}\\
\lfloor KJ/128\rfloor&\text{ for the fifth nesting levels.}
\end{align*}
Thus, for the computation of $v^{p-t}$ we need
\[
\left\lfloor\frac{KJ}{2^{2+l}}\right\rfloor
\]
operations for the nesting level $l\in\{1,2,3,4,5\}$ in stage $p-t+2$.

For every $q\in\{0,1,\dots,t-2\}$ we should determine the number of operations for the nesting level $l\in\{1,2,3,4,5\}$ in stage $p-t+2+q$.
We have already considered the case $q=0$.
By the inductive assumption, under the computation of $v^{p-t}$ we have used the following number of operations
\[
\left\lfloor\frac{KJ}{2^{2+l}\times32^q}\right\rfloor
\]
for the nesting level $l\in\{1,2,3,4,5\}$ in stage $p-t+2+q$.

Since we use the doubling scheme for the computation of $v^{p-t}$ again, then in stage $p-t+2+(q+1)$ we must execute the following number of operations:
\begin{align*}
\left\lfloor KJ/(8\times32^q\times32)\right\rfloor&\text{ for the first nesting level;}\\
\left\lfloor KJ/(16\times32^q\times32)\right\rfloor&\text{ for the second nesting level;}\\
\left\lfloor KJ/(32\times32^q\times32)\right\rfloor&\text{ for the third nesting level;}\\
\left\lfloor KJ/(64\times32^q\times32)\right\rfloor&\text{ for the fourth nesting level;}\\
\left\lfloor KJ/(128\times32^q\times32)\right\rfloor&\text{ for the fifth nesting level.}
\end{align*}
Hence, for the computation of $v^{p-t}$ we need
\[
\left\lfloor\frac{KJ}{2^{2+l}\times32^{q+1}}\right\rfloor
\]
operations for the nesting level $l\in\{1,2,3,4,5\}$ in stage $p-t+2+(q+1)$.

Finally, we make an important remark.
Our considerations make no sense if the number operations is less than $1$.
Indeed, we have the least number of operations
\[
\left\lfloor\frac{KJ}{2^{2+5}\times32^{p-3}}\right\rfloor=\left\lfloor\frac{KJ}{2^{5p-8}}\right\rfloor
\]
for the nesting level $5$ under the computation of $v^1$.
Since we have $ KJ\geq2^{5p-8}$ by the assumption of Statement 3, our considerations are correct.
\end{enumerate}
\end{proof}

Now we want to consider the cases when the computation of $v^1$ stops in stage $p\in\{2,3\}$.
It should be noted that stopping in stage 1 is impossible by Remark \ref{r:p12}.
\begin{lemma}\label{l:p23}
Suppose that we have executed $p$ stages of the process to perform the $Q$-effective implementation of the algorithm $\mathcal{J}$.
\begin{enumerate}
\item
If we stop computing of $v^1$ in stage $p=2$, then
\begin{gather*}
KJ\in\{1,2,\dots,7\}
\intertext{ and the width of the algorithm $\mathcal{J}$ is equal to}
P_{\mathcal{J}}(K,J)=5\cdot KJ.
\end{gather*}
\item
If we stop computing of $v^1$ in stage $p=3$, then
\begin{gather*}
KJ\in\{8,9,\dots,255\}
\intertext{ and the width of the algorithm $\mathcal{J}$ is equal to}
P_{\mathcal{J}}(K,J)=5\cdot KJ+\left\lfloor\frac{KJ}{8}\right\rfloor=5\cdot KJ+t,
\end{gather*}
where $KJ\in\{8t,8t+1,\dots,8t+7\}$ for $t\in\{1,2,\dots,31\}$.
\end{enumerate}
In particular, the width $P_{\mathcal{J}}(K,J)$ is a strictly increasing function of $KJ$ for $KJ\in\{1,\dots,255\}$.
\end{lemma}
\begin{proof}
The process for performing the $Q$-effective implementation of the algorithm $\mathcal{J}$ is described in the proof of Lemma \ref{l:nol+}.
\begin{enumerate}
\item
We execute
\[
KJ,\left\lfloor\frac{KJ}{2}\right\rfloor,\left\lfloor\frac{KJ}{4}\right\rfloor
\]
operations according to the nesting level.

Let $KJ\neq1$.
Then the stopping criterion is that one of the numbers $\lfloor KJ/2\rfloor$, or $\lfloor KJ/4\rfloor$ is equal to $1$, i.e.,
\begin{align*}
\left\lfloor\frac{KJ}{2}\right\rfloor=1&\iff2\leq KJ\leq3;\\
\left\lfloor\frac{KJ}{4}\right\rfloor=1&\iff4\leq KJ\leq7.
\end{align*}

It is clear that the width of the algorithm $\mathcal{J}$ is equal to
\[
P_{\mathcal{J}}(K,J)=5\cdot KJ.
\]
\item
There are
\[
\left\lfloor\frac{KJ}{8}\right\rfloor,\left\lfloor\frac{KJ}{16}\right\rfloor,\left\lfloor\frac{KJ}{32}\right\rfloor,\left\lfloor\frac{KJ}{64}\right\rfloor,\left\lfloor\frac{KJ}{128}\right\rfloor
\]
operations according to the nesting levels.

The stopping criterion is that one of these numbers is equal to $1$, i.e.,
\begin{align*}
\left\lfloor\frac{KJ}{8}\right\rfloor=1&\iff8\leq KJ\leq15;&\left\lfloor\frac{KJ}{16}\right\rfloor=1&\iff16\leq KJ\leq31;\\
\left\lfloor\frac{KJ}{32}\right\rfloor=1&\iff32\leq KJ\leq63;&\left\lfloor\frac{KJ}{64}\right\rfloor=1&\iff64\leq KJ\leq127;\\
\left\lfloor\frac{KJ}{128}\right\rfloor=1&\iff128\leq KJ\leq255.
\end{align*}
Hence, for every $KJ\in\{8,9,\dots,255\}$
\[
\left\lfloor\frac{KJ}{8}\right\rfloor=t\iff KJ\in\{8t,8t+1,\dots,8t+7\}.
\]

So, the width of the algorithm $\mathcal{J}$ is equal to
\[
P_{\mathcal{J}}(K,J)=5\cdot KJ+\left\lfloor\frac{KJ}{8}\right\rfloor=5\cdot KJ+t
\]
if $KJ\in\{8t,8t+1,\dots,8t+7\}$ for $t\in\{1,2,\dots,31\}$.
\end{enumerate}
\end{proof}

\begin{table}[ht]\label{t:p23}
\caption{The width of the algorithm $\mathcal{J}$ for $KJ\in\{1,2,\dots,255\}$}
\fbox{\parbox{300pt}{
\begin{description}
\item[For $KJ\in\{1,2,\dots,7\}$]
we have
\[
P_{\mathcal{J}}(K,J)=5KJ.
\]
\item[For $KJ\in\{8,9,\dots,255\}$]
we have
\[
P_{\mathcal{J}}(K,J)=5\cdot KJ+t
\]
if $KJ\in\{8t,8t+1,\dots,8t+7\}$ for $t\in\{1,2,\dots,31\}$.
\end{description}
}
}
\end{table}

\begin{remark}
Consequently, we fully studied the width of the algorithm $\mathcal{J}$ for $KJ\in\{1,2,\dots,255\}$.
Thus, we should consider the case when $KJ\geq256$.
In addition, by Lemma \ref{l:p23}, we can suppose that we execute $p\geq4$ stages of the process to perform the $Q$-effective implementation of the algorithm $\mathcal{J}$.
\end{remark}

\begin{lemma}\label{l:p4}
Suppose we have executed $p$ stages of the process to perform the $Q$-effective implementation of the algorithm $\mathcal{J}$.

If the stop for computing $v^1$ occurs in stage $p\geq 4$, then
\begin{gather*}
KJ\in\{8\times32^{p-3},8\times32^{p-3}+1,\dots,256\times32^{p-3}-1\}
\intertext{ and the width of the algorithm $\mathcal{J}$ is equal to}
P_{\mathcal{J}}(K,J)=5\cdot KJ +\left\lfloor\frac{KJ}{8}\right\rfloor+\left\lfloor\frac{KJ}{8\times32}\right\rfloor+\dots+\left\lfloor\frac{KJ}{8\times32^{p-3}}\right\rfloor.
\end{gather*}
\end{lemma}
\begin{proof}
The process for performing the $Q$-effective implementation of the algorithm $\mathcal{J}$ is described in the proof of Lemma \ref{l:nol+}.

If we stop computing of $v^1$ in stage $p$, then for the computation of the expression $v^1$ we need to execute the following numbers of operations in stage $p$:
\begin{align*}
\lfloor KJ/(8\times32^{p-3})\rfloor&\text{ for the first nesting level;}\\
\lfloor KJ/(16\times32^{p-3})\rfloor&\text{ for the second nesting level;}\\
\lfloor KJ/(32\times32^{p-3})\rfloor&\text{ for the third nesting level;}\\
\lfloor KJ/(64\times32^{p-3})\rfloor&\text{ for the fourth nesting level;}\\
\lfloor KJ/(128\times32^{p-3})\rfloor&\text{ for the fifth nesting level.}
\end{align*}
The stopping criterion is that one of these numbers is equal to $1$, i.e.,
\begin{align*}
\left\lfloor\frac{KJ}{8\times32^{p-3}}\right\rfloor=1&\iff8\times32^{p-3}\leq KJ\leq16\times32^{p-3}-1;\\
\left\lfloor\frac{KJ}{16\times32^{p-3}}\right\rfloor=1&\iff16\times32^{p-3}\leq KJ\leq32\times32^{p-3}-1;\\
\left\lfloor\frac{KJ}{32\times32^{p-3}}\right\rfloor=1&\iff32\times32^{p-3}\leq KJ\leq64\times32^{p-3}-1;\\
\left\lfloor\frac{KJ}{64\times32^{p-3}}\right\rfloor=1&\iff64\times32^{p-3}\leq KJ\leq128\times32^{p-3}-1;\\
\left\lfloor\frac{KJ}{128\times32^{p-3}}\right\rfloor=1&\iff128\times32^{p-3}\leq KJ\leq256\times32^{p-3}-1.
\end{align*}
Hence, for every $KJ\in\{8\times32^{p-3},8\times32^{p-3}+1,\dots,256\times32^{p-3}-1\}$ we obtain
\begin{multline*}
\left\lfloor\frac{KJ}{8\times32^{p-3}}\right\rfloor=t\iff\\
\iff KJ\in\{(8\times32^{p-3})t,(8\times32^{p-3})t+1,\dots,(8\times32^{p-3})(t+1)-1\}.
\end{multline*}

Under the process for performing the $Q$-effective implementation of the algorithm $\mathcal{J}$, the computation of $v^1$ stops first.
Further stages of the computations have the same number of operations.
So, we can only consider the computation of $v^1$.

Consequently, the width of the algorithm $\mathcal{J}$ is equal to
\[
P_{\mathcal{J}}(K,J)=5\cdot KJ+\left\lfloor\frac{KJ}{8}\right\rfloor+\left\lfloor\frac{KJ}{8\times32}\right\rfloor+\dots+\left\lfloor\frac{KJ}{8\times32^{p-3}}\right\rfloor.
\]

This completes  the proof of Lemma \ref{l:p4}.
\end{proof}

\begin{proposition}\label{p:wJg}
Suppose that $KJ\geq256$.
More exactly,
\[
KJ\in\{8\times32^{p-3},8\times32^{p-3}+1,\dots,256\times32^{p-3}-1\}
\]
for a suitable $p\in\{4,5,\dots\}$.
Furthermore,
\[
KJ=a+8b_0+(8\times32)b_1+(8\times32^2)b_2+\dots+(8\times32^{p-4})b_{p-4}+(8\times32^{p-3})b_{p-3},
\]
where
\begin{align*}
a&\in\{0,1,\dots,7\},\\
b_i&\in\{0,1,\dots,31\}\text{ for every }i\in\{0,1,\dots,p-4\},\\
b_{p-3}&\in\{1,2,\dots,31\}.
\end{align*}
Then the width of the algorithm $\mathcal{J}$ is equal to
\begin{equation}
\label{wJg}
P_{\mathcal{J}}(K,J)=5\cdot KJ+c_0+c_1+\dots+c_{p-3},
\end{equation}
where for every $i\in\{0,1,\dots,p-3\}$
\[
c_i=b_i+32b_{i+1}+\dots+32^{p-3-i}b_{p-3}.
\]

Also, we have the following recurrence formula:
\begin{align*}
c_{p-3}&=b_{p-3}
\intertext{and for every $i\in\{p-3,p-4,\dots,1\}$}
c_{i-1}&=b_{i-1}+32c_i.
\end{align*}
\end{proposition}
\begin{proof}
We apply Lemma \ref{l:p4}.
We should compute
\[
\left\lfloor\frac{KJ}{8}\right\rfloor,\left\lfloor\frac{KJ}{8\times32}\right\rfloor,\dots,\left\lfloor\frac{KJ}{8\times32^{p-3}}\right\rfloor.
\]
Then
\begin{align*}
\left\lfloor\frac{KJ}{8}\right\rfloor&=c_0=b_0+32b_1+32^2b_2+\dots+32^{p-4}b_{p-4}+32^{p-3}b_{p-3},
\intertext{similarly, we get for every $i\in\{1,\dots,p-3\}$}
\left\lfloor\frac{KJ}{8\times32^i}\right\rfloor&=c_i=b_i+32b_{i+1}+\dots+32^{p-3-i}b_{p-3};
\intertext{in particular, we have}
\left\lfloor\frac{KJ}{8\times32^{p-3}}\right\rfloor&=c_{p-3}=b_{p-3}.
\end{align*}

Furthermore, it is obvious that
\[
c_{i-1}=b_{i-1}+32c_i
\]
for every $i\in\{p-3,p-4,\dots,1\}$.
\end{proof}

\begin{remark}
From Proposition \ref{p:wJg} we obtain that the width of the algorithm $\mathcal{J}$ depends only on the product $KJ$.
Therefore, we can write $P_{\mathcal{J}}(KJ)$ instead of $P_{\mathcal{J}}(K,J)$ without loss of generality.
\end{remark}

\begin{corollary}\label{c:inc}
Suppose that for $p\geq4$
\[
KJ=a+8b_0+(8\times32)b_1+(8\times32^2)b_2+\dots+(8\times32^{p-4})b_{p-4}+(8\times32^{p-3})b_{p-3},
\]
where
\begin{align*}
a&\in\{0,1,\dots,7\},\\
b_i&\in\{0,1,\dots,31\}\text{ for every }i\in\{0,1,\dots,p-4\},\\
b_{p-3}&\in\{1,2,\dots,31\}.
\end{align*}
Then
\begin{multline*}
P_{\mathcal{J}}(KJ+1)=\\
=P_{\mathcal{J}}(KJ)+
\begin{cases}
5&\text{if }a\in\{0,1,\dots,6\};\\
6&\text{if }a=7\text{ and }b_0\in\{0,1,\dots,30\};\\
l+6&\text{if }a=7,\ b_0=\dots=b_{l-1}=31,\\
&b_l\in\{0,1,\dots,30\}\text{ for some }l\in\{1,\dots,p-3\};\\
p+4&\text{if }a=7\text{ and }b_0=\dots=b_{p-3}=31.
\end{cases}
\end{multline*}
In particular, the width $P_{\mathcal{J}}(KJ)$ is a strictly increasing function of $KJ$.
\end{corollary}
\begin{proof}
There are four cases.
\begin{enumerate}
\item
The first case is $a\in\{0,1,\dots,6\}$.
Then for every $i\in\{0,1,\dots,p-3\}$ we have
\[
\left\lfloor\frac{KJ+1}{8\times32^i}\right\rfloor=\left\lfloor\frac{KJ}{8\times32^i}\right\rfloor.
\]
By Proposition \ref{p:wJg}
\[
P_{\mathcal{J}}(KJ+1)=P_{\mathcal{J}}(KJ)+5.
\]
\item
In the second case, we have
\[
a=7\text{ and }b_0\in\{0,1,\dots,30\}.
\]
Then
\begin{gather*}
\left\lfloor\frac{KJ+1}{8}\right\rfloor=\left\lfloor\frac{KJ}{8}\right\rfloor+1
\intertext{and for every $i\in\{1,2\dots,p-3\}$}
\left\lfloor\frac{KJ+1}{8\times32^i}\right\rfloor=\left\lfloor\frac{KJ}{8\times32^i}\right\rfloor.
\end{gather*}
By Proposition \ref{p:wJg}
\[
P_{\mathcal{J}}(KJ+1)=P_{\mathcal{J}}(KJ)+5+1=P_{\mathcal{J}}(KJ)+6.
\]
\item
In the third case, we have
\[
a=7,\ b_0=\dots=b_{l-1}=31\text{ and }b_l\in\{0,1,\dots,30\}
\]
for some $l\in\{1,\dots,p-3\}$.
Then
\begin{align*}
KJ&=7+31\left(8+\dots+(8\times32^{l-1})\right)+\\
&\hspace*{19pt}+(8\times32^l)b_l+\dots+(8\times32^{p-3})b_{p-3},\\
KJ+1&=0+0\left(8+\cdots+(8\times32^{l-1})\right)+\\
&\hspace*{19pt}+(8\times32^l)(b_l+1)+\cdots+(8{\times}32^{p-3})b_{p-3}=\\
&=(8\times32^l)(b_l+1)+\dots+(8\times32^{p-3})b_{p-3}.
\end{align*}
By Proposition \ref{p:wJg}
\begin{align*}
P_{\mathcal{J}}(KJ)&=5\cdot KJ+c_0+c_1+\dots+c_{p-3},
\intertext{where for every $i\in\{l,\dots,p-3\}$}
c_i&=b_i+32b_{i+1}+\dots+32^{p-3-i}b_{p-3}
\intertext{and for every $i\in\{0,\dots,l-1\}$}
c_i&=31+32\cdot31+\dots+32^{l-1-i}\cdot31+32^{l-i}b_l+\dots\\
&\hspace*{19pt}+32^{p-3-i}b_{p-3}=\\
&=31(1+32+\dots+32^{l-1-i})+32^{l-i}(b_l+\dots+32^{p-3-l}b_{p-3})=\\
&=31\cdot\frac{32^{l-i}-1}{32-1}+32^{l-i}(b_l+\dots+32^{p-3-l}b_{p-3})=\\
&=(32^{l-i}-1)+32^{l-i}(b_l+\dots+32^{p-3-l}b_{p-3}).
\end{align*}
Again, by Proposition \ref{p:wJg}
\begin{align*}
P_{\mathcal{J}}(KJ+1)&=5\cdot(KJ+1)+c_0'+c_1'+\dots+c_{p-3}',
\intertext{where for every $i\in\{l+1,\dots,p-3\}$}
c_i'&=b_i+32b_{i+1}+\dots+32^{p-3-i}b_{p-3}=c_i,\\
c_l'&=(b_l+1)+32b_{l+1}+\dots+32^{p-3-l}b_{p-3}=c_l+1,
\intertext{and for every $i\in\{0,\dots,l-1\}$}
c_i'&=0+32\cdot0+\dots+32^{l-1-i}\cdot0+32^{l-i}(b_l+1)+\dots\\
&\hspace*{19pt}+32^{p-3-i}b_{p-3}=\\
&=32^{l-i}+32^{l-i}(b_l+\dots+32^{p-3-l}b_{p-3})=c_i+1.
\end{align*}
So, we obtain
\begin{align*}
P_{\mathcal{J}}(KJ+1)&=5\cdot(KJ+1)+(c_0'+\dots+c_{l-1}')+c_l'+(c_l'+\dots+c_{p-3}')=\\
&=(5KJ+5)+\left((c_0+1)+\dots+(c_{l-1}+1)\right)+(c_l+1)+\\
&+(c_l+\dots+c_{p-3})=\\
&=\left(5KJ+c_0+c_1+\dots+c_{p-3}\right)+5+(l+1)=\\
&=P_{\mathcal{J}}(KJ)+l+6.
\end{align*}
\item
Finally, we consider the case
\[
a=7\text{ and }b_0=\dots=b_{p-3}=31.
\]
Then
\begin{align*}
KJ&=7+31\left(8+\dots+(8\times32^{p-3})\right)=\\
&=7+(31\cdot8)\left(1+\dots+32^{p-3}\right)=\\
&=7+(31\cdot8)\cdot\frac{32^{p-2}-1}{32-1}=7+8\cdot(32^{p-2}-1)=8\times32^{p-2}-1,\\
KJ+1&=8\times32^{p-2}.
\end{align*}
By Proposition \ref{p:wJg}
\begin{align*}
P_{\mathcal{J}}(KJ)&=5\cdot KJ+c_0+c_1+\dots+c_{p-3},
\intertext{for every $i\in\{0,\dots,p-3\}$}
c_i&=31+32\cdot31+\dots+32^{p-3-i}\cdot31=\\
&=31(1+32+\dots+32^{p-3-i})=\\
&=31\cdot\frac{32^{p-2-i}-1}{32-1}=32^{p-2-i}-1.
\end{align*}
So, we have
\begin{align*}
P_{\mathcal{J}}(KJ)&=5\cdot(8\times32^{p-2}-1)+(32^{p-2}-1)+\dots+(32-1)=\\
&=40\cdot32^{p-2}-5-(p-2)+32\left(1+\dots+32^{p-3}\right)=\\
&=40\cdot32^{p-2}-p-3+32\cdot\frac{32^{p-2}-1}{32-1}=\\
&=40\cdot32^{p-2}-p-3+31\cdot\frac{32^{p-2}-1}{31}+\frac{32^{p-2}-1}{31}=\\
&=40\cdot32^{p-2}-p-3+(32^{p-2}-1)+\frac{32^{p-2}-1}{31}=\\
&=41\cdot32^{p-2}-p-4+\frac{32^{p-2}-1}{31}.
\end{align*}
Again, by Proposition \ref{p:wJg}
\begin{align*}
P_{\mathcal{J}}(KJ+1)&=5\cdot(KJ+1)+c_0'+c_1'+\dots+c_{p-2}',
\intertext{for every $i\in\{0,\dots,p-2\}$}
c_i'&=0+32\cdot0+\dots+32^{p-3-i}\cdot0+32^{p-2-i}\cdot1=32^{p-2-i}.
\end{align*}
So, we obtain
\begin{align*}
P_{\mathcal{J}}(KJ+1)&=5\cdot(8\times32^{p-2})+32^{p-2}+32^{p-3}+\dots+32^0=\\
&=41\cdot32^{p-2}+\frac{32^{p-2}-1}{31}=P_{\mathcal{J}}(KJ)+p+4.
\end{align*}
\end{enumerate}
\end{proof}

\subsubsection{The width of the Jacobi method for $KJ=2^s$}\label{sss:w2}

Here we consider the case when $KJ=2^s$.
The width formula of the algorithm $\mathcal{J}$ from Proposition \ref{p:wJg} is quite complicated and inconvenient to compute in some cases.

At the same time, it is obvious that the computations are simpler if $KJ=2^s$.
\begin{lemma}\label{l:nol}
Suppose that we execute stage $p$ of the process to perform the $Q$-effective implementation of the algorithm $\mathcal{J}$.
\begin{enumerate}
\item
If $p\geq1$, then under the computation of 
\[
V_p=\left\{u_{kj}^p\mid k\in\{1,\dots, K\},\,j\in\{1,\dots,J\}\right\} 
\]
we must execute the following number of operations:
\begin{align*}
4KJ&=2^{2+s}\text{ for the first nesting level;}\\
2KJ&=2^{1+s}\text{ for the second nesting level;}\\
KJ&=2^s\text{ for the nesting levels $3$, $4$, and $5$.}
\end{align*}
\item
If $p\geq2$ and $s\geq2$, then under the computation of $v^{p-1}$ we must execute the following number of operations:
\begin{align*}
KJ&=2^s\text{ for the nesting levels $1$, $2$, and $3$;}\\
KJ/2&=2^{s-1}\text{ for the fourth nesting level;}\\
KJ/4&=2^{s-2}\text{ for the fifth nesting level.}
\end{align*}
\item
Assume that $p\geq3$ and $s\geq5p-8$.
Let $l\in\{1,2,3,4,5\}$ be a nesting level.
Then for every $t\geq2$ under the computation of $v^{p-t}$  we must execute the following number of operations:
\[
\frac{KJ}{2^{2+l}\times32^{t-2}}=2^{s-l-5p+13}.
\]
\end{enumerate}
In particular, if $p\geq3$ and $s\geq5p-8$, then we use
\[
\frac{KJ}{2^{2+l}\times32^{p-3}}=2^{s-l-5t+8}
\]
operations on the nesting level $l\in\{1,2,3,4,5\}$ to compute $v^1$.
\end{lemma}
\begin{proof}
This lemma is a particular case of Lemma \ref{l:nol+}.
Since for every $q\in\{0,1,\dots,s\}$ we get
\[
\left\lfloor\frac{KJ}{2^q}\right\rfloor=2^{s-q},
\]
then the proof is trivial.
\end{proof}

\begin{lemma}\label{l:1nl}
Suppose that $KJ=2^s$ and we have executed $p\geq3$ stages of the process for performing the $Q$-effective implementation of the algorithm $\mathcal{J}$.
\begin{enumerate}
\item
For each stage, we have the most number of operations on the first nesting level, and each of the other nesting levels has fewer operations.
\item
Furthermore, on the first nesting level, the total number of operations is
\[
P(2^s,p)=5\cdot2^s+2^{s+12-5p}\cdot\frac{2^{5p-10}-1}{31}.
\]
\end{enumerate}
\end{lemma}
\begin{proof}
The first statement is a direct consequence of Lemma \ref{l:nol}.

Indeed, we get
\begin{equation}
\begin{aligned}\label{no}
P(2^s,p)&=4KJ+KJ+\frac{KJ}{8}+\dots+\frac{KJ}{8\times32^{p-3}}=5\cdot2^s+\sum_{i=0}^{p-3}\frac{2^s}{8\times32^i}=\\
&=5\cdot2^s+\sum_{i=0}^{p-3}2^{s-3-5i}=5\cdot2^s+\sum_{i=0}^{p-3}2^{s-3-5i}=\\
&=5\cdot2^s+2^{s+12-5p}\sum_{i=0}^{p-3}2^{s-3-5i-s-12+5p}=\\
&=5\cdot2^s+2^{s+12-5p}\sum_{i=0}^{p-3}2^{5(p-3-i)}=5\cdot2^s+2^{s+12-5p}\sum_{l=0}^{p-3}2^{5l}=\\
&=5\cdot2^s+2^{s+12-5p}\cdot\frac{2^{5(p-2)}-1}{2^5-1}=5\cdot2^s+2^{s+12-5p}\cdot\frac{2^{5p-10}-1}{31}.
\end{aligned}
\end{equation}
\end{proof}

\begin{lemma}\label{l:finp}
Suppose that $KJ=2^s\geq8$.
Also
\[
s\equiv r\pmod{5}\text{ for the suitable }r\in\{-2,-1,0,1,2\}.
\]
In the process of performing the $Q$-effective implementation of the algorithm $\mathcal{J}$, the computation of $v^1$ stops in stage
\[
p=2+\frac{s-r}{5}=2+\left[\frac{s}{5}\right]\geq3,
\]
and on the first nesting level of stage $p$ we have
\[
P(2^s,p)=5\cdot2^s+2^{2+r}\cdot\frac{2^{s-r}-1}{31}.
\]
operations.
\end{lemma}
\begin{proof}
Under the computation of $v^1$, there exists a stage $p$ such that on one of its nesting levels we have only one operation.
After that, the computation of $v^1$ finishes.

Let us describe the process of computing a $Q$-term $v^1$ at the necessary stage $p$.
By Lemma \ref{l:nol}, we obtain
\[
2^{s+12-5p},2^{s+11-5p},2^{s+10-5p},2^{s+9-5p}, 2^{s+8-5p}
\]
operations according to the nesting levels.
The stopping criterion is that one of these numbers is equal to $1$, i.e.,
\[
2^{s+10-r-5p}=1
\]
for some $r\in\{-2,-1,0,1,2\}$.
Hence,
\[
s=5(p-2)+r\iff p=\frac{s+10-r}{5}=2+\frac{s-r}{5}=2+\left[\frac{s}{5}\right]\geq3,
\]
and
\[
s\equiv r\pmod{5}.
\]

From this we obtain the total number of operations on the first nesting level by the equality (\ref{no}):

\begin{align*}
P(2^s,p)&=5\cdot2^s+2^{s+12-5p}\cdot\frac{2^{5p-10}-1}{31}=\\
&=5\cdot2^s+2^{s+12-s-10+r}\cdot\frac{2^{s+10-r-10}-1}{31}=5\cdot2^s+2^{2+r}\cdot\frac{2^{s-r}-1}{31}
\end{align*}
for $r\in\{-2,-1,0,1,2\}$ and $s\equiv r\pmod{5}$.
\end{proof}

\begin{proposition}\label{p:wJ2}
Suppose that $KJ=2^s\geq8$.
The width of the algorithm $\mathcal{J}$ is equal to
\begin{equation}
\label{wJ2}
P_{\mathcal{J}}(2^s)=5\cdot2^s+2^{2+r}\cdot\frac{2^{s-r}-1}{31}
\end{equation}
for $r\in\{-2,-1,0,1,2\}$ and $r\equiv s\pmod{5}$.
\end{proposition}
\begin{proof}
As in the proof of Lemma \ref{l:p4}, we can only consider the computation of $v^1$.

By Lemma \ref{l:finp}, the computation of $v^1$ stops at the stage
\[
p=2+2+\left[\frac{s}{5}\right]\geq3.
\]
On the first nesting level of this stage, we obtain the largest number of operations by Lemma \ref{l:1nl}.

Thus, the width of the algorithm $\mathcal{J}$ is equal to
\[
P_{\mathcal{J}}(2^s)=P(2^s,p)=5\cdot2^s+2^{2+r}\cdot\frac{2^{s-r}-1}{31}
\]
for $r\in\{-2,-1,0,1,2\}$ such that
\[
r\equiv s\pmod{5}.
\]
\end{proof}

\begin{remark}\label{r:p3}
It is interesting to consider the case $p=3$.
Obviously,
\[
p=3\iff s\in\{3,4,5,6,7\}.
\]
Then
\[
s-r=5
\]
for $r\in\{-2,-1,0,1,2\}$ such that $r\equiv s\pmod{5}$.

It follows that the width of the algorithm $\mathcal{J}$ is equal to
\begin{align*}
P_{\mathcal{J}}(2^s)&=5\cdot2^s+2^{2+r}\cdot\frac{2^5-1}{31}=5\cdot2^s+2^{2+r}=\\
&=\begin{cases}
5\cdot8+2^0=41&\text{ for }2^s=8;\\
5\cdot16+2^1=82&\text{ for }2^s=16;\\
5\cdot32+2^2=164&\text{ for }2^s=32;\\
5\cdot64+2^3=328&\text{ for }2^s=64;\\
5\cdot128+2^4=656&\text{ for }2^s=128.
\end{cases}
\end{align*}
\end{remark}

Remark \ref{r:p3} is easily generalized as a corollary of Proposition \ref{p:wJ2}.
\begin{corollary}
Suppose that $KJ=2^{5q-2}\geq8$ and the width of the algorithm $\mathcal{J}$ is equal to
\[
P_{\mathcal{J}}(2^{5q-2})=5\cdot2^{5q-2}+\frac{2^{5q}-1}{31}.
\]
Then
\[
P_{\mathcal{J}}(2^{5q-2+t})=2^t\cdot P_{\mathcal{J}}(2^{5q-2})
\]
for $t\in\{0,1,2,3,4\}$.
\end{corollary}
\begin{proof}
It is a direct generalization of the consideration in Remark \ref{r:p3}.
\end{proof}

By this Corollary we easy make the following table.
\begin{table}[ht]\label{t:2}
\caption{The width of the algorithm $\mathcal{J}$ for $KJ=2^s$}
\begin{gather*}
\begin{array}{|c|c|c|c|}\hline
s&0&1&2\\ \hline
P_{\mathcal{J}}(2^s)&5&10&20\\ \hline
\end{array}\\
\begin{array}{|c|c|c|c|c|c|}\hline
s&3&4&5&6&7\\ \hline
P_{\mathcal{J}}(2^s)&41&82&164&328&656\\ \hline
\end{array}\\
\begin{array}{|c|c|c|c|c|c|}\hline
s&8&9&10&11&12\\ \hline
P_{\mathcal{J}}(2^s)&1313&2626&5252&10504&21008\\ \hline
\end{array}
\end{gather*}
\end{table}

Now we estimate the width of the algorithm $\mathcal{J}$.
More exactly, we would like to estimate the fraction
\[
\frac{2^{s-r}-1}{31}
\]
in the equality (\ref{wJ2}).

\begin{corollary}\label{c:est}
Suppose that $s\geq8$, $r\equiv s\pmod{5}$ for $r\in\{-2,-1,0,1,2\}$.
Also
\[
d(s)=41\cdot2^{s-3}.
\]
Then
\begin{equation}\label{e:estfrac}
2^{s-r-5}<\frac{2^{s-r}-1}{31}<2^{s-r-4}-1.
\end{equation}
Furthermore,
\[
\begin{aligned}
41\cdot2^{s-3}&<P_{\mathcal{J}}(2^s)<21\cdot2^{s-2}-2^{2+r}\leq21\cdot2^{s-2}-1<6\cdot2^s\iff\\
\iff d(s)&<P_{\mathcal{J}}(2^s)<d(s)+2^{s-3}-2^{2+r}\leq d(s)+2^{s-3}-1<6\cdot2^s.
\end{aligned}
\]
\end{corollary}
\begin{proof}
We note the following very trivial result.
If $a>b>c>0$, then
\[
\frac{a}{b}<\frac{a-c}{b-c}.
\]
Therefore, for every $q\in\{2,3,\dots\}$
\[
\frac{2^{5q}-1+1}{31+1}<\frac{2^{5q}-1}{31}<\frac{2^{5q}-1-15}{31-15}\iff2^{5(q-1)}<\frac{2^{5q}-1}{31}<2^{5q-4}-1.
\]

By Proposition \ref{p:wJ2}
\[
P_{\mathcal{J}}(2^s)=5\cdot2^s+2^{2+r}\cdot\frac{2^{s-r}-1}{31}.
\]
How in Remark \ref{r:p3}, we get
\[
s-r=5q\text{ and }q\geq2
\]
for $s\geq8$.
Hence,
\[
2^{s-r-5}<\frac{2^{s-r}-1}{31}<2^{s-r-4}-1.
\]
Now
\begin{multline*}
5\cdot2^s+2^{2+r}\cdot2^{s-r-5}<P_{\mathcal{J}}(2^s)<5\cdot2^s+2^{2+r}(2^{s-r-4}-1)\iff\\
\iff5\cdot2^s+2^{s-3}<P_{\mathcal{J}}(2^s)<5\cdot2^s+2^{s-2}-2^{2+r}\leq5\cdot2^s+2^{s-2}-1\iff\\
\iff2^{s-3}(40+1)<P_{\mathcal{J}}(2^s)<2^{s-2}(20+1)-2^{2+r}\leq2^s(20+1)-1\iff\\
\iff41\cdot2^{s-3}<P_{\mathcal{J}}(2^s)<21\cdot2^{s-2}-2^{2+r}\leq21\cdot2^{s-2}-1.
\end{multline*}
We have
\[
21\cdot2^{s-2}-1<24\cdot2^{s-2}=6\cdot2^s.
\]
Also
\[
21\cdot2^{s-2}=42\cdot2^{s-3}=41\cdot2^{s-3}+2^{s-3}=d(s)+2^{s-3},
\]
then the proof is complete.
\end{proof}

\subsubsection{The width estimations of the Jacobi method for the general case}\label{sss:est}

\begin{remark}
The width formula of the algorithm $\mathcal{J}$ from Proposition \ref{p:wJg} is rather complicated and not always convenient for computations.
At the same time, the width formula of the algorithm $\mathcal{J}$ from Proposition \ref{p:wJ2} is more suitable for computations.
We will try to obtain the width estimations for the algorithm $\mathcal{J}$ for the general case using the formula from Proposition \ref{p:wJ2}.
By Lemma \ref{l:p23} and Table  1 
(p.~\pageref{t:p23}) we can consider the case $KJ\geq257$.

So, we \emph{fix the integer $s$ such that
\[
2^{s-1}< KJ<2^s
\]
for the suitable $s\geq9$.}

Then by Lemma \ref{l:p4}
\[
P_{\mathcal{J}}(2^{s-1})<P_{\mathcal{J}}(KJ)<P_{\mathcal{J}}(2^s).
\]
\end{remark}

\begin{lemma}\label{l:finp(s-1)}
Suppose that $KJ=2^{s-1}\geq2^8$.
Also
\[
s\equiv r\pmod{5}\text{ for the suitable }r\in\{-2,-1,0,1,2\}.
\]
Under performing the $Q$-effective implementation of the algorithm $\mathcal{J}$, the computation of $v^1$ stops at the stage
\begin{align*}
p&=2+\frac{s-r}{5}=2+\left[\frac{s}{5}\right]\text{ for }r\in\{-1,0,1,2\};\\
p-1&=1+\frac{s-r}{5}=1+\left[\frac{s}{5}\right]\text{ for }r=-2.
\end{align*}

If $r\in\{-1,0,1,2\}$, then on the first nesting level of stage $p$ we have
\[
P(2^{s-1},p)=5\cdot2^{s-1}+2^{1+r}\cdot\frac{2^{s-r}-1}{31}
\]
operations.

In the case $r=-2$, then on the first nesting level of stage $p-1$ we have
\[
P(2^{s-1},p-1)=5\cdot2^{s-1}+2^4\cdot\frac{2^{s-3}-1}{31}
\]
operations.
\end{lemma}
\begin{proof}
We should consider the stopping criterion when $KJ=2^{s-1}$.
By Lemma \ref{l:finp},
\[
s-1=5(q-2)+t\text{ for }q=2+\left[\frac{s-1}{5}\right]\text{ and }t\in\{-2,-1,0,1,2\}.
\]
Hence,
\[
5(q-2)+t+1=s=5(p-2)+r\iff t+1-r=5(p-q).
\]
Since
\[
-2+1-2\leq t+1-r\leq 2+1+2\iff-3\leq t+1-r\leq 5,
\]
then
\[
t+1-r\in\{0,5\}.
\]
More exactly,
\begin{align*}
t&=r-1\text{ and }q=p\text{ for }r\in\{-1,0,1,2\};\\
t&=2\text{ and }q=p-1\text{ for }r=-2.
\end{align*}

So, we obtain two cases.
\begin{description}
\item[$q=p$.]
By the equality (\ref{no})
\begin{align*}
P(2^{s-1},p)&=5\cdot2^{s-1}+2^{(s-1)+12-5q}\cdot\frac{2^{5q-10}-1}{31}=\\
&=5\cdot2^{s-1}+2^{s+11-5p}\cdot\frac{2^{5p-10}-1}{31}=\\
&=5\cdot2^{s-1}+2^{s+11-s-10+r}\cdot\frac{2^{s+10-r-10}-1}{31}=\\
&=5\cdot2^{s-1}+2^{1+r}\cdot\frac{2^{s-r}-1}{31}.
\end{align*}
\item[$q=p-1$.]
In this case $r=-2$ and $t=2$, and by the equality (\ref{no}) we obtain
\[
P(2^{s-1},p-1)=5\cdot2^{s-1}+2^{2+2}\cdot\frac{2^{(s-1)-2}-1}{31}=5\cdot2^{s-1}+2^4\cdot\frac{2^{s-3}-1}{31}.
\]
\end{description}
\end{proof}

\begin{lemma}\label{l:dif}
Suppose that
\begin{align*}
s&\equiv r\pmod{5}\text{ for the suitable }r\in\{-2,-1,0,1,2\},\\
p&=2+\frac{s-r}{5}=2+\left[\frac{s}{5}\right].
\end{align*}
If $r\in\{-1,0,1,2\}$, then
\[
P(2^s,p)-P(2^{s-1},p)=5\cdot2^{s-1}+2^{1+r}\cdot\frac{2^{s-r}-1}{31}.
\]

In the case $r=-2$, we have on the first nesting level of stage $p-1$
\[
P(2^s,p)-P(2^{s-1},p-1)=5\cdot2^{s-1}+2^4\cdot\frac{2^{s-3}-1}{31}.
\]
\end{lemma}
\begin{proof}
Let $r\in\{-1,0,1,2\}$.
By the equality (\ref{no}) and Lemma \ref{l:finp(s-1)} we get
\begin{align*}
P(2^s,p)-P(2^{s-1},p-1)&=5\cdot2^s+2^{2+r}\cdot\frac{2^{s-r}-1}{31}-\\
&\hspace*{19pt}-5\cdot2^{s-1}-2^{1+r}\cdot\frac{2^{s-r}-1}{31}=\\
&=5\cdot2^{s-1}+(2^{2+r}-2^{1+r})\frac{2^{s-r}-1}{31}=\\
&=5\cdot2^{s-1}+2^{1+r}\cdot\frac{2^{s-r}-1}{31}.
\end{align*}
Moreover, we can compute more exactly
\[
P(2^s,p)-P(2^{s-1},p)=5\cdot2^{s-1}+
\begin{cases}
8\cdot\dfrac{2^{s-2}-1}{31}&\text{if }r=2;\\[5pt]
4\cdot\dfrac{2^{s-1}-1}{31}&\text{if }r=1;\\[5pt]
2\cdot\dfrac{2^s-1}{31}&\text{if }r=0;\\[5pt]
\dfrac{2^{s+1}-1}{31}&\text{if }r=-1.
\end{cases}
\]

If $r=-2$, then
\begin{align*}
P(2^s,p)-P(2^{s-1},p-1)&=5\cdot2^s+2^0\cdot\frac{2^{s+2}-1}{31}-5\cdot2^{s-1}-2^4\cdot\frac{2^{s-3}-1}{31}=\\
&=5\cdot2^{s-1}+\frac{2^{s+2}-1-2^{s+1}+16}{31}=\\
&=5\cdot2^{s-1}+\frac{2^{s+1}+15}{31}=\\
&=5\cdot2^{s-1}+\frac{16\cdot(2^{s-3}-1)+31}{31}=\\
&=5\cdot2^{s-1}+1+16\cdot\frac{2^{s-3}-1}{31}.
\end{align*}
\end{proof}

\begin{lemma}\label{l:difner}
Suppose that
\begin{align*}
s&\equiv r\pmod{5}\text{ for the suitable }r\in\{-2,-1,0,1,2\},\\
p&=2+\frac{s-r}{5}=2+\left[\frac{s}{5}\right],\\
d(s-1)&=41\cdot2^{s-4}\text{ (cf. Corollary {\rm\ref{c:est}})}.
\end{align*}
If $r\in\{-1,0,1,2\}$, then
\begin{multline*}
41\cdot2^{s-4}<P(2^{s-1},p)-P(2^{s-1},p)<21\cdot2^{s-3}-2^{1+r}\leq\\
\leq21\cdot2^{s-3}-1<6\cdot2^{s-1}{\iff}\\
{\iff}d(s-1)<P(2^{s-1},p)-P(2^{s-1},p)<d(s-1)+2^{s-4}-2^{1+r}.
\end{multline*}
In the case $r=-2$, we have on the first nesting level of stage $p-1$
\[
\begin{aligned}
41\cdot2^{s-4}&<P(2^s,p)-P(2^{s-1},p-1)<21\cdot2^{s-3}-16\iff\\
\iff d(s-1)&<P(2^{s-1},p)-P(2^{s-1},p)<d(s-1)+2^{s-4}-16.
\end{aligned}
\]
\end{lemma}
\begin{proof}
Let $r\in\{-1,0,1,2\}$.
By the inequality (\ref{e:estfrac}) from Corollary \ref{c:est} and Lemma \ref{l:dif}
\begin{multline*}
\hspace*{-10pt}5\cdot2^{s-1}+2^{1+r}\cdot2^{s-r-5}<P(2^{s-1},p)-P(2^{s-1},p)<5\cdot2^{s-1}+2^{1+r}\cdot(2^{s-r-4}-1)\iff\\
\iff5\cdot2^{s-1}+2^{s-4}<P(2^{s-1},p)-P(2^{s-1},p)<5\cdot2^{s-1}+2^{s-3}-2^{1+r}\iff\\
\iff2^{s-4}(5\cdot2^3+1)<P(2^{s-1},p)-P(2^{s-1},p)<2^{s-3}(5\cdot2^2+1)-2^{1+r}\iff\\
\iff41\cdot2^{s-4}<P(2^{s-1},p)-P(2^{s-1},p)<21\cdot2^{s-3}-2^{1+r}\leq21\cdot2^{s-3}-1.
\end{multline*}
Also
\[
21\cdot2^{s-3}-1<24\cdot\cdot2^{s-3}=6\cdot2^{s-1}.
\]

In the case $r=-2$, we obtain by the inequality (\ref{e:estfrac}) and Lemma \ref{l:dif}
\begin{multline*}
5\cdot2^{s-1}+2^4\cdot2^{s-3-5}<P(2^s,p)-P(2^{s-1},p-1)<5\cdot2^{s-1}+2^4(2^{s-3-4}-1)\iff\\
5\cdot2^{s-1}+2^{s-4}<P(2^s,p)-P(2^{s-1},p-1)<5\cdot2^{s-1}+2^{s-3}-16\iff\\
41\cdot2^{s-4}<P(2^s,p)-P(2^{s-1},p-1)<21\cdot2^{s-3}-16.
\end{multline*}
\end{proof}

\begin{proposition}\label{p:gest}
Let $s\geq9$ be the integer such that
\[
2^{s-1}<KJ<2^s.
\]
Also, suppose that
\begin{align*}
s&\equiv r\pmod{5}\text{ for the suitable }r\in\{-2,-1,0,1,2\},\\
p&=2+\frac{s-r}{5}=2+\left[\frac{s}{5}\right],\\
d(s-1)&=41\cdot2^{s-4}\text{ (cf. Corollary {\rm\ref{c:est}})}.
\end{align*}
Let $r\in\{-1,0,1,2\}$.
Then
\begin{gather*}
P(2^{s-1},p)=5\cdot2^{s-1}+2^{1+r}\cdot\frac{2^{s-r}-1}{31}< P_{\mathcal{J}}(KJ)<\\
<5\cdot2^s+2^{2+r}\cdot\frac{2^{s-r}-1}{31}=P(2^s,p),
\intertext{in addition,}
P(2^s,p)-P(2^{s-1},p)=5\cdot2^{s-1}+2^{1+r}\cdot\frac{2^{s-r}-1}{31},\\
41\cdot2^{s-4}<P(2^{s-1},p)-P(2^{s-1},p)<\\
<21\cdot2^{s-3}-2^{1+r}\leq21\cdot2^{s-3}-1<6\cdot2^{s-1}\iff\\
\iff d(s-1)<P(2^{s-1},p)-P(2^{s-1},p)<d(s-1)+2^{s-4}-2^{1+r}.
\end{gather*}

In the case $r=-2$, we have
\begin{gather*}
P(2^{s-1},p-1){=}5\cdot2^{s-1}+\frac{2^{s+1}-16}{31}< P_{\mathcal{J}}(KJ)<5\cdot2^s+\frac{2^{s+2}-1}{31}{=}P(2^s,p),
\intertext{in addition,}
P(2^s,p)-P(2^{s-1},p-1)=5\cdot2^{s-1}+2^4\cdot\frac{2^{s-3}-1}{31},\\
41\cdot2^{s-4}<P(2^{s-1},p)-P(2^{s-1},p)<21\cdot2^{s-3}-16<6\cdot2^{s-1}\iff\\
\iff d(s-1)<P(2^{s-1},p)-P(2^{s-1},p)<d(s-1)+2^{s-4}-16.
\end{gather*}
\end{proposition}
\begin{proof}
The statement of Proposition is hold by Lemmas \ref{l:p4}, \ref{l:finp(s-1)}--\ref{l:difner}.
\end{proof}

\subsection{Final Theorem}

\begin{theorem}\label{th:T2}
Suppose that $\mathcal{S}$ (the scalar product), $\mathcal{G}$ (the Gauss--Jordan) and $\mathcal{J}$ (the Jacobi) are previously analyzed algorithms.
\begin{enumerate}
\item
The $Q$-effective implementation of each algorithm $\mathcal{S}$, $\mathcal{G}$, and $\mathcal{J}$ is realizable.
\item
For the height and width of these algorithms, we have the following.
\begin{description}
\item[Scalar Product.]
For the algorithm $\mathcal{S}$
\[
D_\mathcal{S}=\left\lceil\log n\right\rceil+1\text{ and }P_\mathcal{S}=n.
 \]
\item[Gauss--Jordan.]
For the described algorithm $\mathcal{G}$
\[
D_\mathcal{G}=3n\text{ and }P_\mathcal{G}\geq\frac{3}{2}(n+1)!
 \]
for $n\geq2$.
\item[Jacobi.]
For the described algorithm $\mathcal{J}$
\[
D_\mathcal{J}=5n_0+3+\lceil\log{KJ}\rceil\text{ if $n_0$ such that the value of $v^{n_0}$ is  \textsf{true}}.
\]

If $KJ\in\{1,2,\dots,255\}$, then the width of $\mathcal{J}$ is equal to
\[
P_\mathcal{J}=5KJ+\lfloor KJ/8\rfloor.
\]
If $KJ\geq256$, then for a suitable $p\geq4$ we have
\[
KJ=a+8b_0+(8\times32)b_1+\dots+(8\times32^{p-3})b_{p-3},
\]
where $a\in\{0,\dots,7\}$, $b_i\in\{0,\dots,31\}$ for every $i\in\{0,1,\dots,p-4\}$, and $b_{p-3}\in\{1,2,\dots,31\}$.
Then the width of the algorithm $\mathcal{J}$ is equal to
\[
P_{\mathcal{J}}=5\cdot KJ+c_0+c_1+\dots+c_{p-3},
\]
where for every $i\in\{0,1,\dots,p-3\}$
\[
c_i=b_i+32b_{i+1}+\dots+32^{p-3-i}b_{p-3}.
\]
\end{description}
\end{enumerate}
\end{theorem}
\begin{proof}
It follows from Propositions \ref{p:s}, \ref{p:G}, and \ref{p:wJg}, Lemmas \ref{l:realJ} and \ref{l:p23}.
\end{proof}

\begin{remark}
Note that for the algorithm $\mathcal{J}$, the case $KJ=2^s$ is considered in \ref{sss:w2}.
In this case, we obtain a simpler value for the width of the algorithm $\mathcal{J}$.
Furthermore, in \ref{sss:est} we estimate the width of the algorithm $\mathcal{J}$ from the width values for powers of 2, see Proposition \ref{p:gest}.
It seems to us that these results can be very useful for estimating the width behavior for large $KJ$.
\end{remark}
 
 \section{The $Q$-system}

In this section, we describe the development of a software system, called the $Q$-system, to study the parallelism resources of numerical algorithms that can be computed and compared.

\subsection{Theoretical foundation of the $Q$-system}

The concept of a $Q$-determinant makes it possible to use the following methods for studying the parallelism resource of numerical algorithms.
\begin{enumerate}
\item
A method for constructing the $Q$-determinant of the algorithm.
\item
A method of obtaining the $Q$-effective implementation of the algorithm.
\item
A method for computing the characteristics of the parallelism resource of the algorithm.
\item
A method for comparing the characteristics of the parallelism resource of algorithms.
\end{enumerate}

\subsubsection{A method of constructing the $Q$-determinant of the algorithm}\label{MCQ}

Flowcharts of algorithms are well known and often used.
Here we construct the $Q$-determinants of the algorithms using flowcharts of the algorithms.

We use the following blocks of the flowchart.
\begin{enumerate}
\item
\textsf{Terminal.}
It displays input from the external environment or output to the external environment.
\item
\textsf{Decision.}
It displays a switching type function that has one input and several alternative outputs, one and only one of which can be activated after computing the condition defined in the block.
\item
\textsf{Process.}
It displays a data processing function.
\item
\textsf{Data.}
\end{enumerate}

\begin{description}
\item[Limitations and Clarifications.]
There are some limitations and clarifications for the used flowcharts.
The flowchart has two terminal blocks, one of which is ``Start'' means the beginning of the algorithm, and the other ``End'' is the end of the algorithm.
The ``Start'' block has one outgoing flowline and has no incoming flowline.
The ``End'' block has one incoming flowline and no outgoing flowline.
The ``decision'' block has two outgoing flowlines, one of which corresponds to the transfer of control if the block condition is true, and the other if it is false.
The condition uses one comparison operation.
Its operands don't contain operations.
The ``process'' block contains one assignment of a variable value.
There is no more than one operation to the right of the assignment sign.
So, block chains are used for groups of operations.
We divide the ``data'' blocks into the blocks ``Input data'' and ``Output data''.
\item[Construction of the $Q$-determinant.]\ 
To construct the $Q$-determinant of the algorithm, we analyze the flowchart for fixed values of the dimension parameters of the algorithmic problem
$\bar{N}=\{\bar{n}_1,\dots,\bar{n}_k\}$ if $N\neq\emptyset$.
We investigate the features of the algorithms in which the $Q$-determinants contain various types of $Q$-terms.

If the $Q$-determinant of the algorithm has unconditional $Q$-terms, then the passage through the flowchart should be carried out sequentially in accordance with the execution of the algorithm.

For every $i\in U$ expressions $w^i$ if $N=\emptyset$ and $w^i(\bar{N})$ if $N\neq\emptyset$ are obtained as values of unconditional $Q$-terms $f_i$ (cf. (\ref{qdd})).
They are formed using the contents of the ``process'' blocks involved in the computation of $f_i$ for every $i\in U$.

Suppose that the $Q$-determinant contains conditional $Q$-terms.
The passage through the block ``decision'' with a condition on the input data generates a branching.

Then each branch is processed separately.
As a result of processing one branch, for every $i\in C$ and suitable $j\in\{1,\dots,l(i)\}$ we obtain the pairs of expressions $(u_j^i,w_j^i)$ if $N=\emptyset$ and $(u_j^i(\bar{N}),w_j^i(\bar{N}))$ if $N\neq\emptyset$ (cf. (\ref{qdd})).

The expressions $u_j^i$ if $N=\emptyset$ and $u_j^i(\bar{N})$ if $N\neq\emptyset$ are obtained by the conditions of the ``decision'' blocks that contain the input data.
The expressions $w_j^i$ if $N=\emptyset$ and $w_j^i(\bar{N})$ if $N\neq\emptyset$ are formed using the contents of the ``process'' blocks.
As soon as the first branch is processed, the flowchart handler returns to the nearest block, where the branching occurred and continues processing with the opposite condition.
After all branches are processed, for all $i\in C$ and $j\in\{1,\dots,l(i)\}$ we obtain the pairs of expressions $(u_j^i,w_j^i)$ if $N=\emptyset$ and $(u_j^i(\bar{N}),w_j^i(\bar{N}))$ if $N\neq\emptyset$ for all conditional $Q$-terms $f_i$, 
where $i\in C$ (cf. (\ref{qdd})).

Conditional infinite $Q$-terms are contained in the $Q$-determinants of iterative numerical algorithms.
By limiting the number of iterations, we can reduce the case of the $Q$-determinant with conditional infinite $Q$-terms to the case of the $Q$-determinant with conditional $Q$-terms.
We will use the parameter $L$ for the number of iterations.
This parameter can take any integer positive values.
We denote the value of the parameter $L$ by $\bar{L}$.
The value $\bar{L}$ determines the length of the $Q$-terms $f_i$ for every $i\in I$.
\end{description}

\subsubsection{A method of obtaining the $Q$-effective implementation of the algorithm}

By the $Q$-effective implementation of the algorithm, we mean the simultaneous computation of the expressions $W$ if $N =\emptyset$ (cf. (\ref{W})) and  the expressions $W(\bar{N})$ if $N\neq\emptyset$ (cf. (\ref{WN})) and operations are executed as soon as their operands are computed.

If $N=\emptyset$ and $I\neq\emptyset$, then instead of the expressions $W$ we use the expressions
\begin{multline}\label{WL}
W(\bar{L})=\bigl\{w^i (i\in U); u_j^i,w_j^i (i\in C,j\in\{1,\dots,l(i)\});\\ 
u_j^i,w_j^i (i\in I, j\in\{1,\dots,\bar{L}\})\bigr\}.
\end{multline}
If $N\neq\emptyset$ and $I\neq\emptyset$, then in place of the expressions $W(\bar{N})$ we use the expressions
\begin{multline}\label{WNL}
W(\bar{N},\bar{L})=\bigl\{w^i(\bar{N}) (i\in U); u_j^i(\bar{N}),w_j^i(\bar{N}) (i\in C, j\in\{1,\dots ,l(i)\});\\
u_j^i(\bar{N}),w_j^i(\bar{N}) (i\in I, j\in\{1,\dots,\bar{L}\})\bigr\}.
\end{multline}
The method of obtaining the $Q$-effective implementation of the algorithm is to compute the nesting levels of the operations included in the expressions:
\begin{enumerate}
\item
$W$ if $N=\emptyset$ and $I=\emptyset$;
\item
$W(\bar{L})$ if $N=\emptyset$ and $I\neq\emptyset$;
\item
$W(\bar{N})$ if $N\neq\emptyset$ and $I=\emptyset$;
\item
$W(\bar{N},\bar{L})$ if $N\neq\emptyset$ and $I\neq\emptyset$.
\end{enumerate}

\subsubsection{A method for computing the characteristics of the parallelism resource of the algorithm}

Computation of the parallelism resource characteristics of the algorithm $\mathcal{A}$ is performed according to the following rules.
\begin{enumerate}
\item
If $N=\emptyset$ and $I=\emptyset$, then the height and width of the algorithm $\mathcal{A}$ are independent of the parameters $N$ and $L$.
The height of the algorithm is denoted by $D_\mathcal{A}$ and is computed by the formula (\ref{D}).
The width of the algorithm is denoted by $P_\mathcal{A}$ and is computed by the formula (\ref{P}).
\item
If $N=\emptyset$ and $I\neq\emptyset$, then the height and width of the algorithm $\mathcal{A}$ depend on the parameter $L$.
The value of the algorithm height is denoted by $D_\mathcal{A}(\bar{L})$ and is computed by the formula
\begin{flalign}
D_\mathcal{A}(\bar{L})&=\max\limits_{w\in W(\bar{L})}T^w.\label{DL}
\end{flalign}
The value of the algorithm width is denoted by $P_\mathcal{A}(\bar{L})$ and is computed by the formula
\begin{flalign}
P_\mathcal{A}(\bar{L})&=\max\limits_{1\leq r\leq D_\mathcal{A}(\bar{L})}\sum_{w\in W(\bar{L})}O_r^w\label{PL}
\end{flalign}
if $O_r^{w}$ is the number of operations of the nesting level $r$ of the expression $w$.
\item
If $N\neq\emptyset$ and $I=\emptyset$, then the height and width of the algorithm $\mathcal{A}$ depend on the parameters of $N$.
The value of the algorithm height is denoted by $D_\mathcal{A}(\bar{N})$ and is computed by the formula (\ref{DN}).
The value of the algorithm width is denoted by $P_\mathcal{A}(\bar{N})$ and is computed by the formula~(\ref{PN}).
\item
If $N\neq\emptyset$ and $I\neq\emptyset$, then the height and width of the algorithm $\mathcal{A}$ depend on the parameters $N$ and $L$.
The value of the algorithm height is denoted by $D_\mathcal{A}(\bar{N},\bar{L})$ and is computed by the formula
\begin{flalign}
D_\mathcal{A}(\bar{N},\bar{L})&=\max\limits_{w(\bar{N})\in W(\bar{N},\bar{L})}T^{w(\bar{N})}.\label{DNL}
\end{flalign}
The value of the algorithm width is denoted by $P_\mathcal{A}(\bar{N},\bar{L})$ and is computed by the formula
\begin{flalign}
P_\mathcal{A}(\bar{N},\bar{L})&=\max\limits_{1\leq r\leq D_\mathcal{A}({\bar{N},\bar{L})}}\sum_{w(\bar{N})\in W(\bar{N},\bar{L})}O_r^{w(\bar{N})}\label{PNL}
\end{flalign}
if $O_r^{w(\bar{N})}$ is the number of operations of the nesting level $r$ of the expression $w(\bar{N})$.
\end{enumerate}

\subsubsection{A method for comparing the characteristics of the parallelism resource of algorithms}

Suppose that the algorithms $\mathcal{A}$ and $\mathcal{B}$ solve the same algorithmic problem.
This method compares the height and width of the algorithms $\mathcal{A}$ and $\mathcal{B}$ if the algorithms satisfy one of the following four sets of conditions.
\begin{description}
\item[$N=\emptyset$ {\rm and} $I=\emptyset$.]
The values $D_\mathcal{A}$, $P_\mathcal{A}$, $D_\mathcal{B}$, and $P_\mathcal{B}$ are defined.
\item[$N=\emptyset$ {\rm and} $I\neq\emptyset$.]
The values of $D_\mathcal{A}(\bar{L})$ and $P_\mathcal{A}(\bar{L})$ are defined if $\bar{L}$ is an element of some set $V_\mathcal{A}(L)$.
The values of $D_\mathcal{B}(\bar{L})$ and $P_\mathcal{B}(\bar{L})$ are defined if $\bar{L}$ is an element of some set $V_\mathcal{B}(L)$.
Besides, $V_\mathcal{A}(L)\cap V_\mathcal{B}(L)\neq\emptyset$.
\item[$N\neq\emptyset$ {\rm and} $I=\emptyset$.]
The values of $D_\mathcal{A}(\bar{N})$ and $P_\mathcal{A}(\bar{N})$ are defined if $\bar{N}$ is an element of some set $V_\mathcal{A}(N)$.
The values of $D_\mathcal{B}(\bar{N})$ and $P_\mathcal{B}(\bar{N})$ are defined if $\bar{N}$ is an element of some set  $V_\mathcal{B}(N)$.
Besides, $V_\mathcal{A}(N)\cap V_\mathcal{B}(N)\neq\emptyset$.
\item[$N\neq\emptyset$ {\rm and} $I\neq\emptyset$.]
The values of $D_\mathcal{A}(\bar{N},\bar{L})$ and $P_\mathcal{A}(\bar{N},\bar{L})$ are defined if $(\bar{N},\bar{L})$ is an element of some set  $V_\mathcal{A}(N,L)$.
The values of $D_\mathcal{B}(\bar{N},\bar{L})$ and $P_\mathcal{B}(\bar{N},\bar{L})$ are defined if $(\bar{N},\bar{L})$ is an element of some set $V_\mathcal{B}(N,L)$.
Besides,\\
$V_\mathcal{A}(N,L)\cap V_\mathcal{B}(N,L)\neq\emptyset$.
\end{description}

By this method, we can perform the following computations.\label{Delta}
\begin{description}
\item[$N=\emptyset$ {\rm and} $I=\emptyset$.]
We find
\begin{align*}
\Delta{D}&=D_\mathcal{A}-D_\mathcal{B},\\
\Delta{P}&=P_\mathcal{A}-P_\mathcal{B}.
\end{align*}
\item[$N=\emptyset$ {\rm and} $I\neq\emptyset$.]
We get
\begin{align*}
\Delta{D}&=\sum_{\bar{L}\in (V_\mathcal{A}(L)\cap V_\mathcal{B}(L))}(D_\mathcal{A}(\bar{L})-D_\mathcal{B}(\bar{L})),\\
\Delta{P}&=\sum_{\bar{L}\in (V_\mathcal{A}(L)\cap V_\mathcal{B}(L))}(P_\mathcal{A}(\bar{L})-P_\mathcal{B}(\bar{L})).
\end{align*}
\item[$N\neq\emptyset$ {\rm and} $I=\emptyset$.]
We find
\begin{align*}
\Delta{D}&=\sum_{\bar{N}\in (V_\mathcal{A}(N)\cap V_\mathcal{B}(N))}(D_\mathcal{A}(\bar{N})-D_\mathcal{B}(\bar{N})),\\
\Delta{P}&=\sum_{\bar{N}\in (V_\mathcal{A}(N)\cap V_\mathcal{B}(N))}(P_\mathcal{A}(\bar{N})-P_\mathcal{B}(\bar{N})).
\end{align*}
\item[$N\neq\emptyset$ {\rm and} $I\neq\emptyset$.]
Finally,
\begin{align*}
\Delta{D}&=\sum_{(\bar{N},\bar{L})\in (V_\mathcal{A}(N,L)\cap V_\mathcal{B}(N,L))}(D_\mathcal{A}(\bar{N},\bar{L})-D_\mathcal{B}(\bar{N},\bar{L})),\\
\Delta{P}&=\sum_{(\bar{N},\bar{L})\in (V_\mathcal{A}(N,L)\cap V_\mathcal{B}(N,L))}(P_\mathcal{A}(\bar{N},\bar{L})-P_\mathcal{B}(\bar{N},\bar{L})).
\end{align*}
\end{description}

As a result, we can make the following conclusions.\label{VDelta}
\begin{description}
\item[$\Delta{D}<0$.]
Then the height of the algorithm $\mathcal{A}$ is less than the height of the algorithm $\mathcal{B}$.
\item[$\Delta{D}=0$.]
Then the algorithms $\mathcal{A}$ and $\mathcal{B}$ have the same height.
\item[$\Delta{D}>0$.]
Then the height of the algorithm $\mathcal{A}$ is greater than the height of the algorithm $\mathcal{B}$.
\item[$\Delta{P}<0$.]
Then the width of the algorithm $\mathcal{A}$ is less than the width of the algorithm $\mathcal{B}$.
\item[$\Delta{P}=0$.]
Then the algorithms $\mathcal{A}$ and $\mathcal{B}$ have the same width.
\item[$\Delta{P}>0$.]
Then the width of the algorithm $\mathcal{A}$ is greater than the width of the algorithm $\mathcal{B}$.
\end{description}

\subsection{A software implementation of the $Q$-system}

The $Q$-system consists of two subsystems: for generating the $Q$-determinants of numerical algorithms and for computing the parallelism resources of numerical algorithms.
The first subsystem realizes a method for constructing the $Q$-determinant of an algorithm based on its flowchart.
Other methods for studying the parallelism resource of numerical algorithms are realized in the second subsystem.

\subsubsection{Subsystem for generating the $Q$-determinants of numerical algorithms}

This subsystem is \textsf{.NET} application in object-oriented programming language \textsf{C\#}.
\textsf{Microsoft Visual Studio} was used as a development environment.
We used the data-interchange format JSON to describe the input and output data of the subsystem.
In accordance with the JSON format, the flowchart is a description of the \textsf{Vertices} blocks and the \textsf{Edges} connections.
The \textsf{Vertices} blocks are identified by the number \textsf{Id}, the type \textsf{Type} and the text content \textsf{Content}.
\textsf{Type} values of blocks \textsf{Vertices}:
\begin{enumerate}
\item[(0)]
$0$ is block ``Start'',
\item
$1$ is block ``End'',
\item
$2$ is block ``process'',
\item
$3$ is block ``decision'',
\item
$4$ is block ``Input data'',\label{blocks}
\item
$5$ is block ``Output data''.
\end{enumerate}
The \textsf{Edges} connectors are determined by the numbers of the starting \textsf{From} and ending \textsf{To} blocks and the type of connector \textsf{Type}.
The \textsf{Type} values of connectors \textsf{Edges} are the following:
\begin{enumerate}
\item[(0)]
$0$ is pass by condition ``false'',
\item
$1$ is pass by condition ``true'',
\item
$2$ is normal connection.
\end{enumerate}

In Fig.~1 there is a flowchart in the JSON format of an algorithm that implements the Gauss--Seidel method for solving systems of linear equations.
We use the following notation in this flowchart: $A$ is a coefficient matrix of a system of linear equations, $X$ is an unknown column,
$B$ is a column of constant terms, $X0$ is an initial approximation, $e$ is an accuracy of computations, \textsf{iterations} is a restriction on the number of iterations.

\begin{figure}[t]
\includegraphics[scale=0.68]{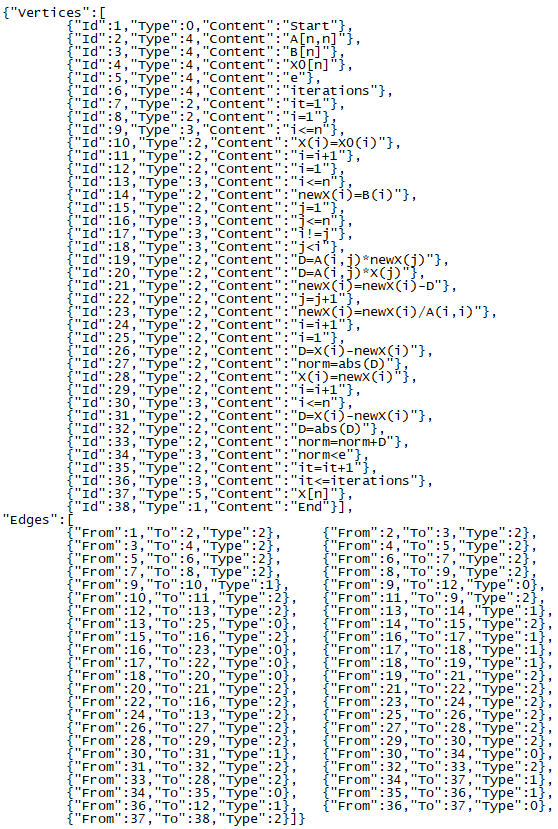}
\caption{The flowchart in the JSON format of an algorithm that implements the Gauss--Seidel method}
\end{figure}

\subsubsection*{Output File}\label{creationQdet}
Let us describe the structure of the output file of this subsystem.
For every $i\in C$ a conditional $Q$-term $f_i$ (cf. (\ref{qdd})) determines $l(i)$ lines of the file, where $l(i)$ is the length of the $Q$-term.
Each of these lines contains the identifier of the output variable computed using this $Q$-term, an equal sign, 
and one pair of expressions $(u_j^i,w_j^i)$ if $N=\emptyset$ or $(u_j^i(\bar{N}),w_j^i(\bar{N}))$ if $N\neq\emptyset$ for some $j\in\{1,\dots,l(i)\}$.
Here, for every $j\in\{1,\dots,l(i)\}$ $Q$-terms, such that $u_j^i$ is an unconditional logical $Q$-term, and $w_j^i$ is an unconditional $Q$-term.
The expressions $u_j^i,w_j^i$ if $N=\emptyset$, or $u_j^i(\bar{N}),w_j^i(\bar{N})$ if $N\neq\emptyset$ are described in the JSON format and separated by semicolons.
The description of the conditional infinite $Q$-term is similar to the description of the conditional $Q$-term, since the length of the conditional infinite $Q$-term is limited by the value of the parameter $L$.
For every $i\in U$ an unconditional $Q$-term $f_i$ (cf. (\ref{qdd})) determines one line of the file.
In this case, there is no logical $Q$-term, so a space is used instead.
Thus, the output file contains the representation of the algorithm in the form of a $Q$-determinant for fixed parameters of dimension $N$ and a limited number of iterations $L$.

In Fig.~2 there is the output file for one iteration of the Gauss--Seidel method if a matrix $A$ is of order 2.
The following notation is used for description of expressions in the JSON format: \textsf{op} is an operation, \textsf{fO} is the first operand (firstOperand), \textsf{sO} is the second operand (secondOperand), \textsf{od} is the operand.

\begin{figure}[t]
\includegraphics[width=\textwidth]{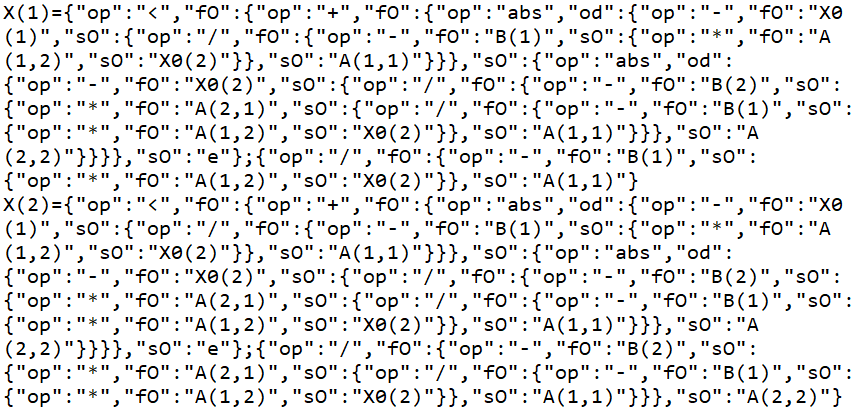}
\caption{The output file of the subsystem for generating the $Q$-determinants of numerical algorithms for the Gauss--Seidel method}
\end{figure}

All variables in the flowchart are divided into four categories:
\begin{enumerate}
\item
dimension parameters of an algorithmic problem,
\item
input,
\item
output,
\item
internal variables that don't belong to the first three categories.
\end{enumerate}
For example, we have the following categories for the flowchart in Fig.~1:
\begin{description}
\item[first category]
is $n$;
\item[second category]
are the variables $\left\{A(i,j), B(i), X0(i) \mid i,j\in\{1,\dots,n\}\right\}$, $e$, \newline\textsf{iterations};
\item[third category]
are the variables $X(i)$ for all $i\in\{1,\dots,n\}$;
\item[fourth category]
are the variables $it$, $i$, $newX(i)$ for all $i\in\{1,\dots,n\}$, $j$, $D$, $norm$.
\end{description}

Further, we describe the work process for each of the categories of variables.
Each of the categories of variables is stored in its own collection.
\begin{description}
\item[First category.]
The user inputs the values of the variables of the first category on request of the program.

The collection of dimension parameters stores the following pair for every dimension parameter: the identifier of dimension parameter and its value.
\item[Second category.]
All variables of the second category, with the exception of the variable \textsf{iterations}, have no values, since the program constructs the $Q$-determinants depending on the identifiers of the input variables.
To limit the number of iterations, the variable \textsf{iterations} is used, and the user inputs its value on request of the program.
This value is assigned to the variable \emph{static int iterations}.

The collection of input variables contains the identifiers of input variables.
\item[Third and fourth categories.]
The values of the variables of the third and fourth categories change when processing the flowchart.

The output and internal variable collections store the following pair for every such variable: the variable identifier and its value.
\end{description}

Now let's move on to processing blocks of the flowchart according to the JSON format (p.~\pageref{blocks}).

We do several steps.
\begin{description}
\item[First.]
We are looking for blocks of type 4.
These blocks include the identifiers of input variables and the dimension parameters of the algorithmic problem.
The identifiers of the dimension parameters are enclosed in square brackets and separated by commas (see. Fig.~1).
They are extracted from blocks of type 4, and the user is prompted to input their values.
The identifier of the dimension parameter and the entered value are written to the dimension parameters collection.
Also the identifiers of input variables are extracted from blocks of type 4.
Indices are added to the identifiers if input variables depend on the dimension parameters.
The resulting input variables are written to the collection of input variables.
If there is an identifier \textsf{iterations}, then the user is prompted to input the number of iterations.

The identifiers of output variables are extracted from blocks of type 5.
Indices are added to the identifiers if output variables depend on the dimension parameters.
The resulting output variables and their values $0$ of type \emph{string} are written to the output variable collection.

For example, suppose the user inputs the value of the dimension parameter $n=2$ under processing a flowchart of an algorithm that implements the Gauss--Seidel method (see. Fig.~1).
Then a pair of $(n,2)$ will be written to the collection of dimension parameters.
The variables $A(i,j)$,  $B(i)$, $X0(i)$ for all $i,j\in\{1,2\}$, $e$, \textsf{iterations} will be written to the collection of input variables.
The pairs $(X(1),0)$ and $(X(2),0)$ will be written to the collection of output variables.
\item[Passage from 0-block to 1-block.]
After that we pass from a block with type 0 to a block with type 1 along the flowchart.
The description of the flowchart determines the order of passing blocks.
Blocks of types 2 and 3 are processed during the passage.
\item[Processing 2-block.]
First, we analyze the right-hand side of the assignment operator when processing a block of type 2.
The following options are possible: a number, a variable of any category with or without indexes, a unary or binary operation.

If there is an operation, then the operands are analyzed.
The following options are possible for operands: a number, a variable of any category with or without indices.
We compute the values of all available indices and determine the category of variables.

Further, we determine the indices values of the variable to the left of the assignment operator, depending on their availability.
Then we should determine the category of the obtained variable.
A variable can be either output or internal.
It is also possible that the variable does not belong to any category.
In this case, we consider it as a new internal variable.

The result of the analysis of a block with type 2 is to compute the value to the right of the assignment operator and write it to the directory as the value of the variable to the left of the assignment operator.
If the value to the right of the assignment operator depends on input variables, then it is formed as an expression in the JSON format.
In this case, we use the same format as the format of expression description in the output file (see. Fig.~2).
After passing through the flowchart, the output variable identifiers and the corresponding unconditional $Q$-terms will be in the output variables catalog.
\item[Processing 3-block.]
When processing a block of type 3, the comparison operation and each operand are determined.
The following options for operands are possible: number, variable of any category with or without indexes.
We compute the values of all obtained indices and determine the category of variables.
If the condition has no input variables, then it is computed and control is transferred to the next block depending on the computed value of $1$ (\textsf{true}), or $0$ (\textsf{false}).
Suppose that the condition contains input variables.
Then, at the first pass through the block, the control is transferred to the next block by the value of $1$.
Also, during the second pass through the block, the control is transferred to the next block by the value of $0$.
Let's say we don't have an input variable \textsf{iterations}.
Under the transfer of the control by the value of $1$, we add the condition as a conjunctive term to the resulting logical $Q$-term.
Under the transfer of the control by the value of $0$, we add the negation of the condition as a conjunctive term to the resulting logical $Q$-term.
If the input variable \textsf{iterations} exists, then the logical $Q$-term is formed from the condition of the last 3-block with input variables only without taking into account \textsf{iterations}.
\item[Branching.]
It arises if we have blocks of type 3 with input variables in the conditions.
We process one branch in one pass along the flowchart.
After processing the branch, the identifiers of output variables and the corresponding pairs of expressions 
\[
(u_j^i,w_j^i)\, (i\in C\cup I)\text{ if }N=\emptyset\text{ or }(u_j^i(\bar{N}),w_j^i(\bar{N}))\, (i\in C\cup I)\text{ if }N\neq\emptyset
\]
for some $j\in\{1,\dots,l(i)\}$ are written to the output file, as shown in Fig.~2.
In practice, sometimes we have to cancel the output of the processing branch results in the output file.
To implement this feature, the internal variable \textsf{empout} is used.
By default, the value of the variable \textsf{empout} is set to $1$, and it is assigned the value $0$ in case if it is necessary to cancel the output to the output file.
\item[Exhaustion of branches.]
Now we describe the procedure for exhausting branches.
A branch is determined by a sequence of 3-blocks with conditions containing input variables, as well as ways to exit blocks by the value of $1$ or $0$.
The branch collection stores information about processed branches.
After processing the next branch, the collection stores the sequence of pairs, consisting of \textsf{Id} branch block numbers and output values from
$1$ or $0$ blocks.
The pairs in the collection are followed in order of processing the blocks.
If we have the first pass through the flowchart, then we record information concerning the first processed branch.
Suppose that the last pair in the branch collection has an output value of $1$ from the block.
Then we change this value to $0$ and get a new pair $P$.
Otherwise, we delete the pairs with the output value $0$ from the last to the pair with the exit from the block equal to $1$.
Deletion ends when there are either no such pairs or a block with an output value of $1$ is found.
If all pairs are deleted, it means that all branches are processed.
In this case, the processing of the flowchart is completed.
If the collection is not empty, then the last pair has a block with an output value of $1$.
We change this value to $0$ and get a new pair $P$.
Now the new branch appears as a subsequence of pairs of the processed branch ending in $P$ and its extension according to the flowchart.
As a result, new items can be written to the branch collection after the pair $P$.
\end{description}

\subsubsection{The subsystem for computing the parallelism resource of numerical algorithms}

The subsystem includes a database, server and client applications.

To develop the database, the database management system \textsf{PostgreSQL} was used.
The database contains two entities: \textsf{Algorithms} and \textsf{Determinants}.
\begin{description}
\item[Entity \textsf{Algorithms}]
has the following attributes.
\begin{enumerate}
\item
A primary key, i.e., an algorithm identifier.
\item
An algorithm name.
\item
Algorithm description.
\item
The number of $Q$-determinants loaded into the database and corresponding to different values of the parameters $N$ and $L$.
\end{enumerate}
\item[Entity \textsf{Determinants}]
has the following attributes.
\begin{enumerate}
\item
A primary key, i.e., an identifier of the $Q$-determinant.
\item
An unique algorithm identifier.
\item
Values of dimension parameters $N$ if $N\neq\emptyset$, otherwise $0$.
\item
The $Q$-determinant according to (\ref{W}), (\ref{WL}), (\ref{WN}), and (\ref{WNL}), where
\begin{description}
\item[(\ref{W})]
if $N=\emptyset$ and $I=\emptyset$;
\item[(\ref{WL})]
if $N=\emptyset$ and $I\neq\emptyset$;
\item[(\ref{WN})]
if $N\neq\emptyset$ and $I=\emptyset$;
\item[(\ref{WNL})]
if $N\neq\emptyset$ and $I\neq\emptyset$.
\end{description}
\item
The value of $D_\mathcal{A}$ according to (\ref{D}), (\ref{DL}), (\ref{DN}), and (\ref{DNL}), where
\begin{description}
\item[(\ref{D})]
if $N=\emptyset$ and $I=\emptyset$,
\item[(\ref{DL})]
if $N=\emptyset$ and $I\neq\emptyset$,
\item[(\ref{DN})]
if $N\neq\emptyset$ and $I=\emptyset$,
\item[(\ref{DNL})]
if $N\neq\emptyset$ and $I\neq\emptyset$.
\end{description}
\item
The value of $P_\mathcal{A}$ according to (\ref{P}), (\ref{PL}), (\ref{PN}), and (\ref{PNL}), where
\begin{description}
\item[(\ref{P})]
if $N=\emptyset$ and $I=\emptyset$,
\item[(\ref{PL})]
if $N=\emptyset$ and $I\neq\emptyset$,
\item[(\ref{PN})]
if $N\neq\emptyset$ and $I=\emptyset$,
\item[(\ref{PNL})]
if $N\neq\emptyset$ and $I\neq\emptyset$.
\end{description}
\item
The value of the parameter $L$ (the number of iterations) if $I\neq\emptyset$, otherwise $0$.
\end{enumerate}
\end{description}

We have developed a server application for interacting with database entities.
This application implements the following methods.
\begin{enumerate}
\item
Recording a new algorithm.
\item
Updating algorithm information.
\item
Obtaining a list of algorithms with complete information about them.
\item
Comparing the characteristics of the parallelism resource of two algorithms that solve the same algorithmic problem.
\item
Removing an algorithm along with its $Q$-determinants.
\item
Loading a new $Q$-determinant and computed characteristics of the parallelism resource.
\item
Obtaining a list of $Q$-determinants.
\item
Downloading a $Q$-determinant.
\item
Removing a $Q$-determinant.
\end{enumerate}

Now we describe the process of interacting with the database.
\begin{description}
\item[Converting of $Q$-determinants for the database.]
The database should\newline
contain the $Q$-determinant of the algorithm in the form of one of the following sets of expressions in the JSON format:
\begin{enumerate}\label{WQD}
\item
$W$ (cf. (\ref{W})) if $N=\emptyset$ and $I=\emptyset$,
\item
$W(\bar{L})$ (cf. (\ref{WL})) if $N=\emptyset$ and $I\neq\emptyset$,
\item
$W(\bar{N})$ (cf. (\ref{WN})) if $N\neq\emptyset$ and $I=\emptyset$,
\item
$W(\bar{N},\bar{L})$ (cf. (\ref{WNL})) if $N\neq\emptyset$ and $I\neq\emptyset$.
\end{enumerate}
We have a description of the representation of the algorithm obtained by the subsystem for generating $Q$-determinants (cf. p.~\pageref{creationQdet}).
Therefore, we have developed special software to convert an initial description into a format for the database.
\item[The $Q$-effective implementation and the parallelism resource.]
Also,\newline
 we have developed and implemented an original algorithm for methods of obtaining the $Q$-effective implementation of the algorithm and computing the characteristics of the parallelism resource of the algorithm.
The algorithm is as follows.
\begin{enumerate}
\item
The input is the $Q$-determinant of the algorithm $\mathcal{A}$ as a set of expressions.
\item
Then we transform the $Q$-determinant into a special data structure.
\item
When creating this data structure, we get an array of lists.
\item
The number of array lists is equal to the number of nesting levels of all operations of the set of expressions of the $Q$-determinant. (cf. p.~\pageref{WQD}).
\item
All operations of one nesting level are contained in one list.
\end{enumerate}
Creating a special data structure is performed as follows.
\begin{enumerate}
\item
For numbers and variables, the algorithm returns an empty array.
\item
If the array of the operand of the unary operation is obtained, the algorithm adds a new list to this array containing this unary operation.
\item
If arrays of operands of a binary operation are obtained, then the algorithm joins these arrays into one array and adds a new list to it containing this binary operation.
\item
The algorithm get its own array of lists for each expression of the set of expressions of the $Q$-determinant.
After that, the algorithm joins the arrays of all expressions of the set of expressions of the $Q$-determinant into one array.
\end{enumerate}
Here, the join of arrays is getting a single array of lists, in which each list is formed as a result of combining lists of source arrays containing operations of the same nesting level.
The resulting array of lists gives an idea of the execution plan of the $Q$-effective implementation of the algorithm, namely, points out the order of operations.
In this case, the value of $D_\mathcal{A}$ is equal to the length of the array, and $P_\mathcal{A}$ is equal to the size of the largest list in the array.
\item[Comparing the characteristics of the parallelism resource.]
We have\newline 
developed and implemented an original algorithm to perform a method for comparing the parallelism resource characteristics of two algorithms.
Several $Q$-determinants of the algorithm can be loaded into the database.
Moreover, in the database, each $Q$-determinant corresponds to its set of values of the dimension parameters $\bar{N}$ and the number of iterations $\bar{L}$ of the algorithm.
In particular, the database contains the value $\bar{N}=0$ if $N=\emptyset$, and the value $\bar{L}=0$ if $I=\emptyset$.
An algorithm for comparing the parallelism resource characteristics of the algorithms $\mathcal{A}$ and $\mathcal{B}$ is as follows.
\begin{enumerate}
\item
The input data is the identifiers of the algorithms $\mathcal{A}$ and $\mathcal{B}$.
\item
The comparison algorithm is performed in three stages.
\begin{enumerate}
\item[(i)]
At the first stage, the identifiers of pairs of $Q$-determinants of the algorithms $\mathcal{A}$ and $\mathcal{B}$ are determined, which correspond to the same values of $\bar{N}$ and $\bar{L}$.
First, for the algorithms $\mathcal{A}$ and $\mathcal{B}$, the identifiers of all $Q$ -determinants loaded into the database are determined.
Then, the algorithm finds the values $\bar{N}$ and $\bar{L}$ corresponding to the obtained identifiers of $Q$-determinants.
Next, pairs of the identifiers of the algorithms $\mathcal{A}$ and $\mathcal{B}$ are analyzed and pairs for which the corresponding values of $\bar{N}$ and $\bar{L}$ are the same.
As a result, pairs of the identifiers of $Q$-determinants of the algorithms $\mathcal{A}$ and $\mathcal{B}$ with the corresponding identical values $\bar{N}$ and $\bar{L}$ will be obtained.
The resulting pairs of the identifiers of $Q$-determinants are written into a two-dimensional array so that the first row of the array contains the identifiers of $Q$-determinants of the algorithm $\mathcal{A}$, and the second row of the array contains the identifiers of $Q$-determinants of the algorithm $\mathcal{B}$.
\item[(ii)]
In the second stage, we determine the data for the comparison.
This is done as follows.
We use the identifiers of $Q$-determinants written in a two-dimensional array from (i).
By these identifiers, the attribute values are determined for comparison, that is, the height values $D_\mathcal{A}(\bar{N},\bar{L})$ and $D_\mathcal{B}(\bar{N},\bar{L})$ or the width values $P_\mathcal{A}(\bar{N},\bar{L})$ and $P_\mathcal{B}(\bar{N},\bar{L})$.
The obtained attribute values are written to another two-dimensional array of the same size into places corresponding to the identifiers of their $Q$-determinants.
\item[(iii)]
The third stage finishes the comparison.
If the two-dimensional data array obtained in the second stage is empty, a message will be displayed that the comparison is not possible.
This situation arises when the algorithms $\mathcal{A}$ and $\mathcal{B}$ don't have $Q$-determinants in the database with the same values 
$\bar{N}$ and $\bar{L}$.
Otherwise, iterates over the columns of a two-dimensional data array.
When iterating over columns in each of the columns, the value of the second row is subtracted from the value of the first row, and the results are summed.
The result will be a value $\Delta{D}$ or $\Delta{P}$ (p.~\pageref{Delta}).
\end{enumerate}
\item
The obtained values of $\Delta{D}$ and $\Delta{P}$ allow us to make a conclusion about the relationship between the characteristics of the parallelism resource of the algorithms $\mathcal{A}$ and $\mathcal{B}$ (p.~\pageref{VDelta}).
\end{enumerate}
\end{description}

The $Q$-system is open for viewing information.
Therefore, to eliminate unauthorized access to information editing, an authentication method was used.

Now we describe the client application.
It manages the database by calling server application methods.
\begin{description}
\item[The main page.]
This page of the client application contains the table with the query results of all the algorithms recorded in the database.
There is also an login interface for authorized users to add, edit and delete algorithms.
Interface elements allow any user to select algorithms for comparison and get the result of their comparison.
In addition, the main page shows the number of $Q$-determinants written to the database for each algorithm.
\item[Access from the main page.]
Hence, we can go to the page containing the query results for all $Q$-determinants for the selected algorithm.
It contains for each $Q$-determinant the values of the dimension parameters $\bar{N}$, the number of iterations $\bar{L}$, the values of the parallelism resource characteristics $D(\bar{N},\bar{L})$ and $P(\bar{N},\bar{L}))$.
The interface enables any user to download any $Q$-determinant into a file.
Authorized users can also use the interface to add and remove $Q$-determinants.
\item[Approximation of the parallelism resource characteristics.]
For this we are developing the functionality of the $Q$-system for the characteristics of the parallelism resource $D(\bar{N},\bar{L})$ and $P(\bar{N},\bar{L})$.
The interface uses the sum sign ($\sum$) to access the developed version of this functionality.
We also test and study the functionality of the $Q$-system for graphical representation of the function obtained by approximating the characteristics of the parallelism resource $D(\bar{N},\bar{L})$ and $P(\bar{N},\bar{L}))$.
\end{description}

\subsection{The results of the trial operation of the $Q$-system}

The $Q$-system is available at
\begin{center}
\texttt{https://qclient.herokuapp.com}
\end{center}

For the trial operation of the $Q$-system, we used numerical algorithms with various structures of $Q$-determinants.
Namely, we consider $Q$-determinants having the following structures.
\begin{description}
\item[$Q$-determinants consisting of unconditional $Q$-terms.]
The following algorithms have such $Q$-determinants.
\begin{enumerate}
\item
Algorithm for computing the scalar product of vectors without using the doubling scheme.
\item
Algorithm for computing the scalar product of vectors with using the doubling scheme.
\item
Dense matrix multiplication algorithm without the doubling scheme.
\item
Dense matrix multiplication algorithm with the doubling scheme.
\end{enumerate}
\item[$Q$-determinants consisting of conditional $Q$-terms.]
The following\newline 
algorithms have such $Q$-determinants.
\begin{enumerate}
\item
An algorithm for implementing the Gauss--Jordan method for solving systems of linear equations.
\item
An algorithm for finding the maximum element in a sequence of numbers.
\item
An algorithm for solving a quadratic equation without using the unary negative ($-$) operation.
\item
An algorithm for solving a quadratic equation using the unary negative ($-$) operation.
\end{enumerate}
\item[$Q$-determinants consisting of conditional infinite $Q$-terms.]
Such $Q$-de\-terminants have the algorithms for implementing the Gauss--Seidel and Jacobi methods for solving systems of linear equations.
\end{description}

We have performed the following steps for each of the algorithms.
\begin{enumerate}
\item
Writing the name and description of the algorithm to the database.
\item
Constructing the flowchart of the algorithm subject to limitations from \ref{MCQ}.
\item
A text description of the flowchart in the JSON format.
\item
Using the first subsystem of the $Q$-system, obtaining representations of the algorithm in the form of a $Q$-determinant for several values of $\bar{N}$ if $N\neq\emptyset$ and for several values of $\bar{L}$ if $I\neq\emptyset$.
The singular representation of the algorithm in the form of a $Q$-determinant was obtained when $N=\emptyset$ and $I=\emptyset$.
\item
Converting the obtained $Q$-determinants to the format for the database and writing them to the database along with the corresponding values of
$\bar{N}$ and $\bar{L}$.
\item
After that, the $Q$-system computed and stored in the database for each $Q$-determinant the values of the parallelism resource characteristics.
\begin{enumerate}
\item
If $N=\emptyset$ and $I=\emptyset$, then the height and width functions of the algorithm are constants.
\item
Otherwise, the height and width functions of the algorithm are stored in a database in a table form.
In this regard, the first version of the functionality of the $Q$-system for approximating the height and width functions was developed.
This enables us to estimate these functions for different values of $\bar{N}$ and $\bar{L}$.
\end{enumerate}
\end{enumerate}

We obtained some practical results during the trial operation of the $Q$-system.
\begin{enumerate}
\item
The $Q$-system shows that the height of the algorithm for implementing the Gauss-Jordan method is $3n$, where $n$ is the matrix size of a system of linear equations.
This estimate confirms Theorem \ref{th:T2} (see p.~\pageref{th:T2}).
\item
The $Q$-system enables us to compare the height and width of any two algorithms that solve the same algorithmic problem.
Comparison of dense matrix multiplication algorithms with the doubling scheme and without the doubling scheme allows us to conclude that the width of the algorithms is the same, but the height of the second algorithm is greater.
Thus, we can suppose that the $Q$-effective implementation of the algorithm with the doubling scheme will be faster.
At the same time, to execute the $Q$-effective implementations of these algorithms, the same amount of computer system resources (computing cores, processors) will be required.
\end{enumerate}

We made a considerable number of improvements and changes to the $Q$-system during its trial operation.
Now we can declare the operation of the system quite reasonable (acceptable) for further use.
It can also be said that the trial operation of the $Q$-system pointed the correctness of the solutions developed for its creation.
Certainly, the $Q$-system will be tested further.

\section{A method of designing parallel programs for the $Q$-effective implementations of algorithms}

\subsection{Method description}

It is well known that the flowchart of a numerical algorithm enables us to develop a sequential program.
Similarly, we can also develop a parallel program using the $Q$-determinant of a numerical algorithm.
This idea is the basis of the proposed method.
We use the following arguments.
\begin{enumerate}
\item
We can construct the $Q$-determinant for any numerical algorithm.
\item
We can describe the $Q$-effective implementation of the algorithm.
\item
We can immediately create the program code if the $Q$-effective implementation of the algorithm is realizable.
\end{enumerate}

The concept model of a $Q$-determinant (see Section \ref{s:Qdet}) are called \emph{the basic model}.
This model allows us to study only machine-independent properties of algorithms.
We extend the basic model to take into account the features of implementing algorithms on real parallel computing systems.
A new model of the concept of a $Q$-determinant is obtained by adding models of parallel computing: PRAM \cite{al:pram} for shared memory and
BSP \cite {al:bsp} for distributed memory.
We use an extended model of the concept of a $Q$-determinant to propose a method for designing parallel programs for the $Q$-effective implementations of algorithms.

The method consists of the following stages.
\begin{enumerate}
\item
Construction of the $Q$-determinant of an algorithm.
\item
Description of the $Q$-effective implementation of the algorithm.
\item
Development of a parallel program for the realizable $Q$-effective implementation of the algorithm.
\end{enumerate}

In the first and second stages of the method, we use the basic model of the concept of a $Q$-determinant, and in the third stage, an extended model.
A program are called \emph{$Q$-effective} if it is designed by this method.
Also, the process of designing a $Q$-effective program will be called $Q$-effective programming.
A $Q$-effective program uses the parallelism resource of the algorithm completely, because it performs the $Q$-effective implementation of the algorithm.
So, it has the highest parallelism among programs that implement the algorithm.
This means that a $Q$-effective program uses more resources of a computer system than programs that perform other implementations of the algorithm, i.e., it is most effective.
A $Q$-effective program can be used for parallel computing systems with shared or distributed memory.

Development of a parallel program for the $Q$-effective implementation of the algorithm is the third stage of our method.
For shared memory, a description of the $Q$-effective implementation of the algorithm for developing a $Q$-effective program is sufficient.
For distributed memory, we should add a description of the distribution of computing between computing nodes.
If we use distributed memory, then our research is limited to the principle of ``master-slave'' \cite{al:massl}.
For example, this principle is often used for cluster computing systems.
We describe the principle of ``master-slave'' as follows.
To compute we use one ``master'' computing node and several ``slave'' computing nodes.
The node ``master'' is denoted by $M$, and the set of ``slave'' nodes by $S$.
The computational process is divided into several supersteps.
\begin{description}
\item[Superstep 1] initialization.
Superstep ends with barrier synchronization.
\item[Superstep 2]
each node of $S$ receives a task for computation from the node $M$.
\item[Superstep 3]
each node of $S$ performs the computation without exchange with other nodes.
\item[Superstep 4]
all nodes of $S$ send the computation results to the node $M$.
Superstep ends with barrier synchronization.
\item[Superstep 5]
in the node $M$, the computation results are combined.
\item[Superstep 6]
completion, conclusion of results.
\end{description}
The supersteps 2--5 can be repeated.

To develop a $Q$-effective program, we should use parallel programming tools.
Now OpenMP technology is most widely used for organizing parallel computing on multiprocessor systems with shared memory.
In this study, this technology are used in the development of $Q$-effective programs for shared memory computing systems.
The MPI technology is the de facto standard for parallel programming on distributed memory.
In this study, this technology are used in the development of $Q$-effective programs for distributed memory computing systems.

\subsection{Examples of application of the method for some algorithms}

We show the use of this method for algorithms that have the $Q$-determinants with $Q$-terms of various types.
In particular, we describe the development features of $Q$-effective programs using distributed memory in the third stage.

\subsubsection{Dense matrix multiplication algorithm with\\ the doubling scheme}

We describe this algorithm by stages of the proposed method.
\begin{description}
\item[Stage 1.]
If we have two matrices
\[
A=\left[a_{ij}\right]_{\substack{i=1,\dots ,n\\j=1,\dots ,k}} \text{ and } B=\left[b_{ij}\right]_{\substack{i=1,\dots ,k\\j=1,\dots ,m}},
\]
then the result is the matrix 
\[
C=\left[c_{ij}\right]_{\substack{i=1,\dots ,n\\j=1,\dots ,m}}, 
\]
where for every $i\in\{1,\dots ,n\}$ and $j\in\{1,\dots ,m\}$
\begin{gather}
c_{ij}=\sum_{s=1}^ka_{is}b_{sj}.\label{qmm}
\end{gather}
So, the algorithm of matrix multiplication is represented in the form of a $Q$-determinant.
The $Q$-determinant consists of $nm$ unconditional $Q$-terms.
\item[Stage 2.]
The $Q$-effective implementation of the algorithm of matrix multiplication is that all $Q$-terms $\sum_{s = 1}^ka_{is}b_{sj}$ for every $i\in\{1,\dots,n\}$ and $j\in\{1,\dots,m\}$ must be computed simultaneously.
Since all multiplication operations are ready to be performed, they must be performed simultaneously.
The result will be the chains that are formed by the addition operation.
The number of these chains is equal to $nm$.
Each chain is computed by the doubling scheme.
Thus, the $Q$-effective implementation of the algorithm is realizable.
\item[Stage 3.]
We got a description of the $Q$-effective implementation in the second stage.
Further, we use this description to develop a $Q$-effective program for shared memory.
To develop a $Q$-effective program for distributed memory, you need to add a description of the distribution of computing between the computing nodes of the computing system.
For this we use the principle of ``master-slave''.
Each $Q$-term $\sum_{s=1}^ka_{is}b_{sj}$ for every $i\in\{1,\dots ,n\}$ and $j\in\{1,\dots ,m\}$ is computed on a separate computational node of $S$.
If the number of nodes of $S$ is less than the number of $Q$-terms, then nodes of $S$ can compute several $Q$-terms.
The nodes of $S$ receive information from the node $M$ to compute $Q$-terms $\sum_{s=1}^ka_{is}b_{sj}$ for every $i\in\{1,\dots ,n\}$ and $j\in\{1,\dots ,m\}$.
The result of computing each $Q$-term is transmitted to the ``master'' node $M$.
\end{description}

Examples of designing $Q$-effective programs for performing the matrix multiplication algorithm are described in \ref{resprom}.

\subsubsection{The Gauss--Jordan method for solving systems of linear equations}

We describe the algorithm for the Gauss--Jordan method also by stages of the proposed method.
\begin{description}
\item[Stage 1.]
Previously, we constructed the $Q$-determinant of the algorithm $\mathcal{G}$ for implementing the Gauss--Jordan method, see p.~\pageref{qd:G} in \ref{sss:G}.
It consists of $n$ conditional $Q$-terms of length $n!$ and
\[
x_j=\{(u_1,w_1^j),\dots ,(u_{n!},w_{n!}^j)\}\text{ for all }j\in\{1,\dots ,n\}
\]
is a representation of the algorithm $\mathcal{G}$ in the form of a $Q$-determinant.
\item[Stage 2.]
By definition of the $Q$-effective implementation all unconditional $Q$-terms
\[
\{u_i,w_i^j\}\text{ for every }i\in\{1,\dots ,n!\}\text{ and }j\in\{1,\dots ,n\}
\]
should be computed simultaneously.
Therefore, two computational process should be carried out at the same time: the parallel computation of matrices
$\bar A^{j_1},\bar A^{j_1j_2},\dots,\bar A^{j_1j_2\dots j_n}$ for all possible values of the numbers $j_1,j_2,\dots ,j_n$,
as well as the parallel computation of $Q$-terms $u_i$ for every $i\in\{1,\dots ,n!\}$.
The leading elements of the matrix for each step of the algorithm are determined by computing $Q$-terms $u_i$ for every $i\in\{1,\dots,n!\}$.
The computation of matrices $\bar A^{j_1},\bar A^{j_1j_2},\dots,\bar A^{j_1j_2\dots j_n}$ and $Q$-terms $u_i$ stops if they do not correspond to the leading elements.
\begin{description}
\item[The first round of computations.]
We begin to compute the matrices $\bar A^{j_1}$ for every ${j_1}\in\{1,\dots ,n\}$ and $Q$-terms $u_i$ for every $i\in\{1,\dots ,n!\}$ simultaneously.
The computations of $u_i$ for every $i\in\{1,\dots ,n!\}$ are started from subexpressions
\[
L_{j_1}\wedge(a_{1j_1}\neq 0)\text{ for every }j_1\in\{1,\dots ,n\},
\]
because only their operations are ready for execution.
There exists a unique integer $j_1\in\{1,\dots ,n\}$ such that a subexpression
\[
L_{j_1}\wedge (a_{1j_1}\neq 0)
\]
has the value \textsf{true}.
Let $r_1$ be such a value of $j_1$.
Further, we finish the computation of $\bar A^{j_1},\bar A^{j_1j_2},\dots ,\bar A^{j_1j_2\dots j_n}$ and $u_i$ for all
$i\in\{1,\dots ,n!\}$ if $j_1$ isn't equal to $r_1$.
\item[The second round of computations.] 
We begin to compute matrices $\bar A^{r_1j_2}$ for every $j_2\in\{1,\dots,n\}\setminus\{r_1\}$ and $Q$-terms $u_i$ for every $i\in\{1,\dots,n!\}$ with $j_1=r_1$ simultaneously.
Under the computation of $u_i$ for every $i\in\{1,\dots,n!\}$ with $j_1=r_1$ we should find the values of subexpressions
\[
L_{j_2}\wedge(a^{r_1}_{2j_2}\neq 0)\text{ for every }j_2\in\{1,\dots ,n\}\setminus\{r_1\},
\]
because only their operations are ready for execution.
There exists a unique integer $j_2\in\{1,\dots ,n\}\setminus\{r_1\}$ such that a subexpression
\[
L_{j_2}\wedge (a^{r_1}_{2j_2}\neq 0)
\]
has the value \textsf{true}.
Let $r_2$ be such a value of $j_2$.
Further, we finish the computation of $\bar A^{r_1j_2},\dots ,\bar A^{r_1j_2\dots j_n}$ and $u_i$ for all $i\in\{1,\dots ,n!\}$ if $j_2$ isn't equal to $r_2$.
\item[The next $n-3$ rounds of computations] are executed similarly.
\item[In conclusion] we need to compute a single matrix
\[
\bar A^{r_1\dots r_{n-1}j_n}\text{ for }j_n\notin\{r_1,r_2,\dots ,r_{n-1}\},
\]
as the parameter $j_n$ has a single value.
Let $r_n$ be such a value of $j_n$.
\end{description}
The result is
\[
x_{r_j}=a^{r_1\dots r_n}_{j,n+1}\text{ for every }j\in\{1,\dots ,n\},
\]
this is the solution to the original system of linear equations.

Thus, the $Q$-effective implementation of the algorithm $\mathcal{G}$ for the Gauss-Jordan method is realizable.
\item[Stage 3.]
We obtained a description of the $Q$-effective implementation of the algorithm $\mathcal{G}$ for the Gauss-Jordan method in the second stage.
Further, we use this description to develop a $Q$-effective program for shared memory.
To develop a $Q$-effective program for distributed memory, you need to add a description of the distribution of computing between the computing nodes of the computing system.
For this we use the principle of ``master-slave''.

We have to compute each of matrices
\[
\left\{\bar A^{j_1}\mid j_1\in\{1,\dots,n\}\right\}
\]
and the corresponding logical $Q$-term $u_i$ for a suitable $i\in\{1,\dots,n!\}$ in its own node of $S$.
If the number of nodes of $S$ is less than $n$, then the nodes of $S$ should perform computations for several values of $j_1$.
The nodes of $S$ get information from the node $M$ to compute matrices $\bar A^{j_1}$ for any $j_1\in\{1,\dots,n\}$ and the corresponding
$Q$-terms $u_i$ for any $i\in\{1,\dots,n!\}$.
The computation results of $r_1$ and $\bar A^{r_1}$ are transferred to the node $M$.

The nodes of $S$ get information from the node $M$ to compute matrices $\bar A^{r_1j_2}$ for every $j_2\in\{1,\dots,n\}\setminus\{r_1\}$ and the corresponding $Q$-terms $u_i$ for every $i\in\{1,\dots,n!\}$.
We have to compute each of matrices
\[
\left\{\bar A^{r_1j_2}\mid j_2\in\{1,\dots,n\}\setminus\{r_1\}\right\}
\]
and the corresponding logical $Q$-term $u_i$ for a suitable $i\in\{1,\dots,n!\}$ in its own node of $S$.
The computation results of $r_1$ and $\bar A^{r_1}$ are transferred to the node $M$.
The next $n-2$ rounds of computations are executed similarly.
\end{description}

Examples of designing $Q$-effective programs for performing the algorithm $\mathcal{G}$ for the Gauss--Jordan method are described in \ref{resprom}.

\subsubsection{The Jacobi method for solving systems of linear equations}\label{ss:JM}

At the end, we describe the algorithm for the Jacobi method according to the stages of the proposed method.
It should be noted that this method differs from the Jacobi method for solving systems of grid equations in the subsection \ref{ss:gJ}.
Although, of course, they are closely related, as we pointed out at the beginning of the subsection \ref{ss:gJ}.
\begin{description}
\item[Stage 1.]
We construct the $Q$-determinant of the algorithm for implementing the Jacobi method for solving a system of linear equations
\begin{gather*}
A\vec{x}=\vec{b}, \text{ where } A=\left[a_{ij}\right]_{i,j=1,\dots ,n}\text{ and }a_{ii}\neq 0\text{ for every }i\in\{1,\dots ,n\}, \\
\vec{x}=(x_1,\dots ,x_n)^T\text{ and }\vec{b}=(a_{1,n+1},\dots ,a_{n,n+1})^T.
\end{gather*}
Let $\vec{x}^0$ be an initial approximation.
Then the iteration process can be written as
\[
x_i^{k+1}=\frac{b_i}{a_{ii}}-\sum_{j\neq i}\frac{a_{ij}}{a_{ii}}x_j^k\text{ for every }i\in\{1,\dots,n\}\text{ and }k\in\{0,1,\dots\}.
\]
A sufficient (but not necessary) condition for the convergence of this method is that the matrix $A$ is strictly diagonally dominant, i.e.,
\[
\left|a_{ii}\right|>\sum _{j\neq i}\left|a_{ij}\right|.
\]
In many cases, the stopping criterion is the condition
\[
\|\vec{x}^{k+1}-\vec{x}^k\|<\epsilon,
\]
where $\epsilon$ is a computation accuracy.

\begin{qdet}
The $Q$-determinant of the algorithm consists of $n$ conditional infinite $Q$-terms.
A representation of the algorithm in the form of a $Q$-determinant is written as
\begin{gather}
x_i=\left\{(\|\vec{x}^1-\vec{x}^0\|<\epsilon,x_i^1),\dots ,(\|\vec{x}^k-\vec{x}^{k-1}\|<\epsilon,x_i^k),\dots\right\}\label{qj}
\end{gather}
for all  $i\in\{1,\dots ,n\}$.
\end{qdet}
\item[Stage 2.]
For convenience, we denote
 \[
 u^l=\|\vec{x}^l-\vec{x}^{l-1}\|<\epsilon
 \]
for every $l\in\{1,2,\dots\}$.
Then a representation of the algorithm in the form of a $Q$-determinant (cf. (\ref{qj})) is
\[
x_i=\left\{(u^1,x_i^1),(u^2,x_i^2),\dots ,(u^k,x_i^k),\dots\right\}
\]
for all $i\in\{1,\dots ,n\}$.

We describe the $Q$-effective implementation of the considered algorithm.
\begin{description}
\item[Firstly,]
we should compute $Q$-terms $x_i^1$ for all $i\in\{1,\dots ,n\}$ simultaneously.
Then we should compute $Q$-terms $u^1$ and $x_i^2$ for all $i\in\{1,\dots ,n\}$ simultaneously.
If the value of $u^1$ is \textsf{true}, then we should finish the computation, since we get the solution with the given accuracy, i.e., $n$-tuple
\[
\left\{x_i=x_i^1\right\}_{i\in\{1,\dots ,n\}}
\]
is the solution of a system of linear equations by the Jacobi method.
\item[Further,]
for $k\geq 2$, we continue to compute $Q$-terms $u^k$ and $x_i^{k+1}$ for all $i\in\{1,\dots,n\}$ simultaneously.
If the value of $u^k$ is \textsf{true}, then we should finish the computation, since we get the solution with the given accuracy, i.e., $n$-tuple
\[
\left\{x_i=x_i^k\right\}_{i\in\{1,\dots ,n\}}
\]
is the solution of a system of linear equations by the Jacobi method.
\end{description}
Thus, the $Q$-effective implementation of the algorithm for the Jacobi method is realizable.
\item[Stage 3.]
We got a description of the $Q$-effective implementation of the algorithm for the Jacobi method in the second stage.
Further, we use this description to develop a $Q$-effective program for shared memory.
To develop a $Q$-effective program for distributed memory, you need to add a description of the distribution of computing between the computing nodes of the computing system.
For this we use the principle of ``master-slave''.

We describe the process of implementing the algorithm for distributed memory.
Every component of the vector $\vec{x}^k$ for any $k\in\{1,2,\dots\}$ is computed on a separate node of $S$.
If the number of nodes of $S$ is less than $n$, then the nodes of $S$ should perform the computations for several components of the vector $\vec{x}^k$ for any $k\in\{1,2,\dots\}$.
\begin{description}
\item[Firstly,]
the nodes of $S$ receive information from the ``master'' node $M$ to compute $Q$-terms $x_i^1$ for every $i\in\{1,\dots ,n\}$.
The result of the computation is $\vec{x}^1$, passed to the node $M$.
\item[After that,]
the nodes of $S$ receive information from the node $M$ to compute $Q$-terms $x_i^2$ for every $i\in\{1,\dots,n\}$.
The node $M$ computes the value of $\|\vec{x}^1-\vec{x}^0\|<\epsilon$.
At the same time, $x_i^2$ for all $i\in\{1,\dots,n\}$ are computed on the nodes of $S$ simultaneously.
\item[The next iterations] are executed similarly.
\end{description}
\end{description}

Examples of designing $Q$-effective programs for performing the algorithm for the Jacobi method are described in \ref{resprom}.

\subsection{Designing $Q$-effective programs}\label{resprom}

Several works contain application of the method for designing $Q$-effective programs.
In these works, there is an experimental study of the developed $Q$-effective programs for algorithms with $Q$-determinants of various structures.
Further, we describe given works in detail.

The research was performed on the supercomputer ``Tornado'' at South Ural State University.
The programming language \textsf{C++} was used in the development of programs.
\begin{description}
\item[Shared memory.]
We used one computing node with 24 computing cores.
For development of programs the technology OpenMP was used.
\item[Distributed memory.]
We used several computing nodes.
The technologies of MPI and OpenMP were used for development of programs.
\end{description}
Developed $Q$-effective programs can use the parallelism resource of algorithms completely.
However, there were insufficient resources of the computer system for use all parallelism under computational experiments.
Under these conditions, the dynamic characteristics of the programs were evaluated.
The program acceleration is computed by the formula
\[
S=\frac{T_1}{T_p},
\]
where $T_1$ is the execution time of the program on one computational core, $T_p$ is the execution time of the program on $p$ computational cores.
The effectiveness of the program was computed by the formula
\[
E=\frac{S}{p},
\]
where $S$ is the acceleration of the program, $p$ is the number of computational cores used.

Multiplication algorithms for dense and sparse matrices were investigated in \cite{al:val}.
The CSR format is used to store sparse matrices.
Fig.~3 and fig.~4 show $Q$-effective programs for the dense matrix multiplication algorithm for shared and distributed memory.
\begin{figure}[h]
\includegraphics[width=\textwidth]{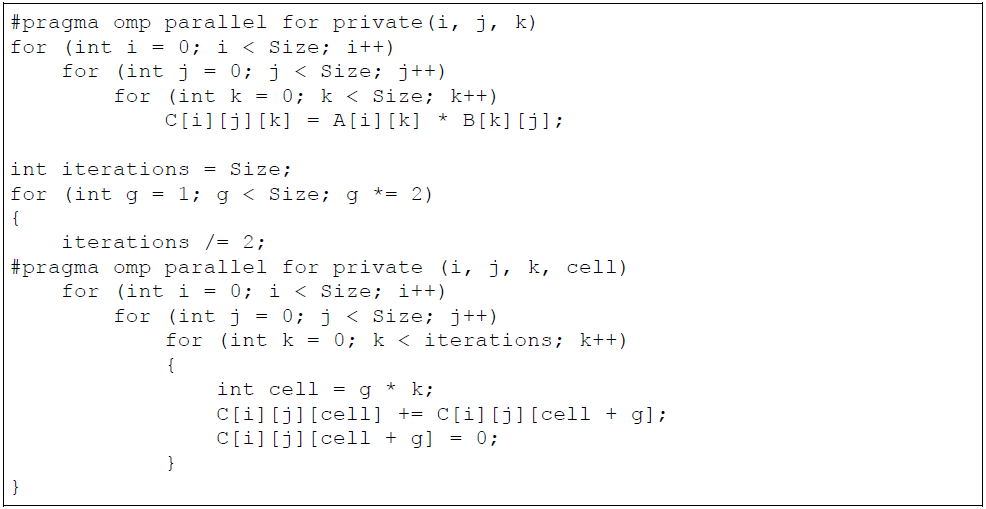}
\caption{The $Q$-effective program for the dense matrix multiplication algorithm for shared memory}
\end{figure}

\begin{figure}[h]
\includegraphics[width=\textwidth]{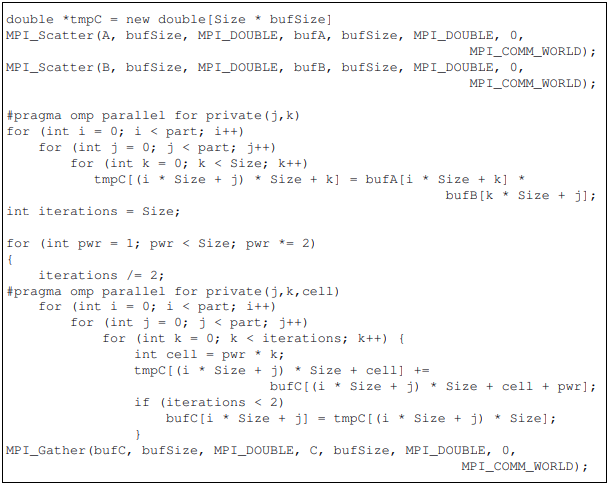}
\caption{The $Q$-effective program for the dense matrix multiplication algorithm for distributed memory}
\end{figure}
Also fig.~5 and fig.~6 show $Q$-effective programs for the sparse matrix multiplication algorithm for shared and distributed memory.
\begin{figure}[t]
\includegraphics[width=\textwidth]{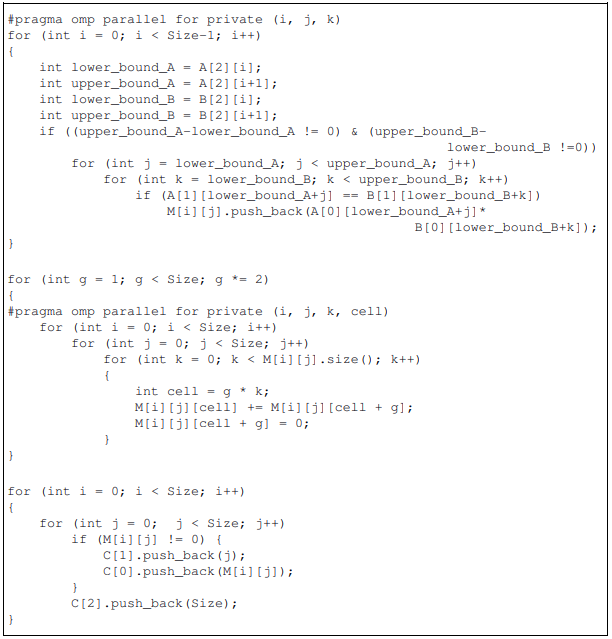}
\caption{The $Q$-effective program for the sparse matrix multiplication algorithm for shared memory}
\end{figure}
\begin{figure}[t]
\vspace{-61pt}
\includegraphics[width=\textwidth]{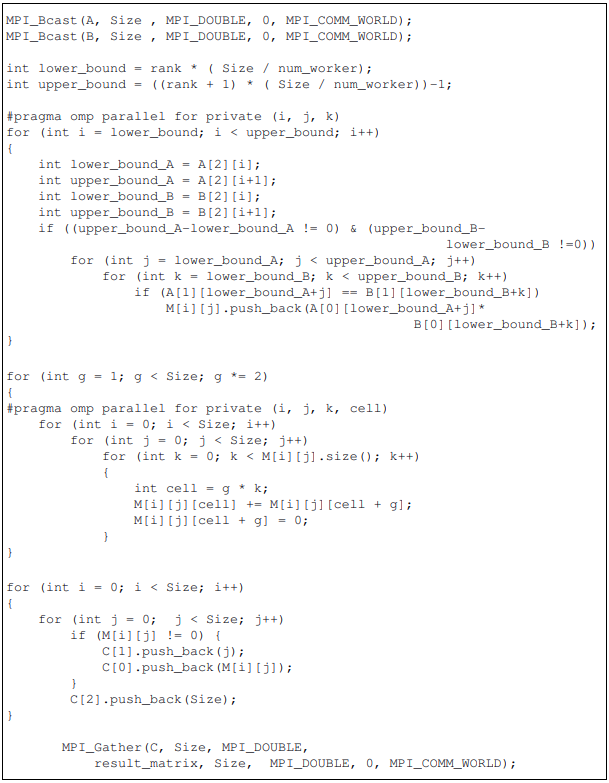}
\caption{The $Q$-effective program for the sparse matrix multiplication algorithm for distributed memory}
\end{figure}
Fig.~7 and fig.~8 show the acceleration of $Q$-effective programs for the dense and sparse matrix multiplication algorithms for shared and distributed memory.
\begin{figure}[t]
\includegraphics[width=\textwidth]{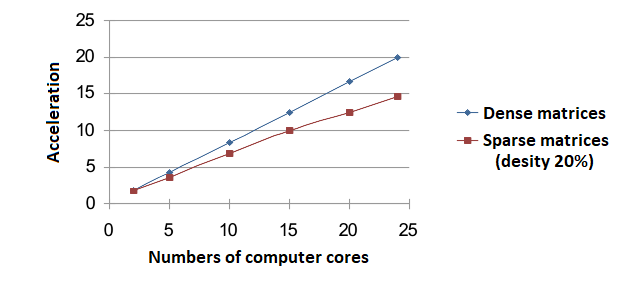}
\caption{The acceleration of $Q$-effective programs for the matrix multiplication algorithm for shared memory}
\end{figure}
\begin{figure}[t]
\includegraphics[width=\textwidth]{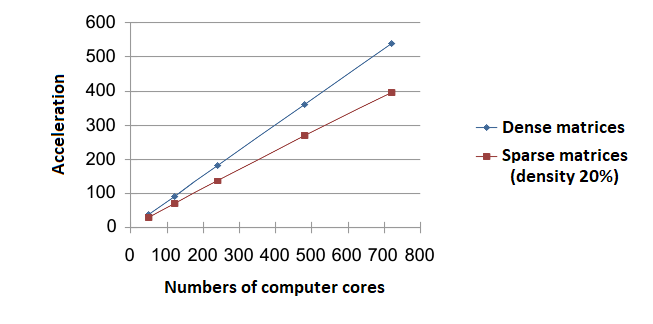}
\caption{The acceleration of $Q$-effective programs for the matrix multiplication algorithm for distributed memory}
\end{figure}
Finally, fig.~9 and fig.~10 show the efficiency of $Q$-effective programs for the dense and sparse matrix multiplication algorithms for shared and distributed memory.
Here the experimental results for matrices of size $30000\times 30000$ are used.
\begin{figure}[t]
\includegraphics[width=\textwidth]{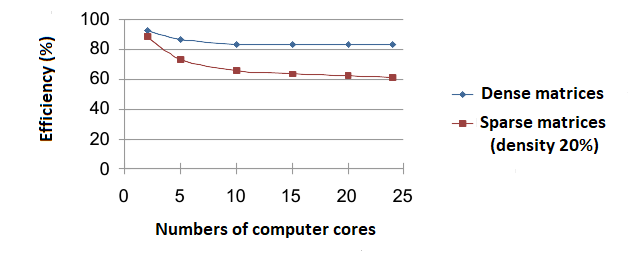}
\caption{The efficiency of $Q$-effective programs for the matrix multiplication algorithm for shared memory}
\end{figure}
\begin{figure}[t]
\includegraphics[width=\textwidth]{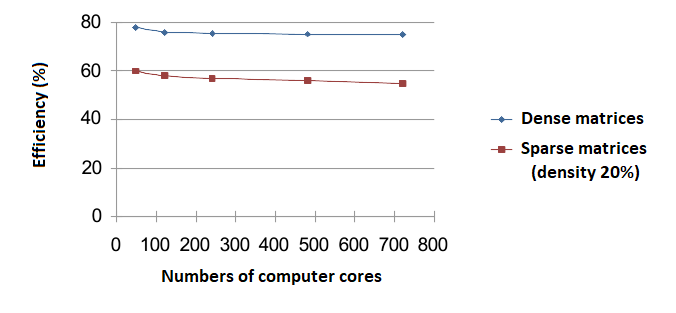}
\caption{The efficiency of $Q$-effective programs for the matrix multiplication algorithm for distributed memory}
\end{figure}

In \cite{al:tarasov} considers the Gauss-Jordan method for solving systems of linear equations.
In \cite{al:lapteva} there were the results of a study of the Jacobi method for solving systems of linear equations.
The method of designing $Q$-effective programs was tested on algorithms with small and large parallelism resources.
For example, in \cite{al:bazhenova}, a sweep method was considered for solving a system of linear three-point equations, which has a small parallelism resource.
Thus, in this case, distributed memory should not be used, because it can lead to reduced performance.
Also, in \cite{al:bazhenova} considered the Fourier method for solving a system of difference equations, which has a large parallelism resource.
Therefore, it can be effectively  implemented on distributed memory by the principle of ``master-slave".
It should be noted the study of the practical application of designing $Q$-effective programs shows that the principle of ``master-slave" is inappropriate for some algorithms, since the interchange between computing nodes increases.
The Jacobi method for solving a system of five-point difference equations \cite{al:kond} is an acceptable example for this statement.

There is also some interest in testing the method of designing $Q$-effective programs for algorithms with complex $Q$-determinants.
For example, such algorithms implement the Gauss--Jordan method \cite{al:tarasov} and the Gauss--Seidel method to solve a system of linear equations \cite{al:nechep}.
These algorithms have complex but well-structured $Q$-determinants.
Many other well-known algorithms have similar properties.
Well-structured $Q$-determinants make it easy to create $Q$-effective programs for these algorithms.

\subsection{Some ideas about designing $Q$-effective programs}

A further development of the model of the $Q$-determinant concept can be the addition of parallel computing models for computing systems with other architectural features.
Also, to create $Q$-effective programs, it is possible to use various programming languages and parallel programming technologies.
So, for one numerical algorithm, there is a potentially infinite set of $Q$-effective programs.
All of them execute the $Q$-effective implementation of the algorithm, this is the fastest implementation of the algorithm from a formal point of view.
It is quite possible, there is not the best of all $Q$-effective programs for given algorithm in terms of performance.
But each of these programs is the most effective in the computing infrastructure for which it was created.

\section{Application of $Q$-effective programming}

Here we consider the application of the mentioned results to increase the efficiency of the implementation of algorithmic problems.
Suppose that for solving an algorithmic problem we can use the set of numerical methods $\mathcal{M}$ with set of parameters $N$.
Also for every $M\in\mathcal{M}$ we can use the class of numerical algorithms $\mathcal{A}(M)$.
Then the process of developing an effective program for solving given algorithmic problem is as follows.
\begin{description}
\item[Choice of the algorithm $A^*(M)\in\mathcal{A}(M)$ for every $M\in\mathcal{M}$] with the minimum value of the height using the $Q$-system.
\item[Choice of the algorithm $A\in\left\{A^*(M)\mid M\in\mathcal{M}\right\}$] with the minimum value of the height.
\item[Development] a $Q$-effective program for the algorithm $A$.
\end{description}
The described process is called \emph{$Q$-effective programming in the broad sense}.
So, in the broad sense the main statements of the technology of $Q$-effective programming are the following.
\begin{description}
\item[Description] of the algorithmic problem.
\item[Description] of the methods for solving the algorithmic problem.
\item[Description] of the classes of numerical algorithms for the implementation of each of the methods.
\item[Investigation] of the parallelism resources of algorithms from the selected\newline 
classes with the help of the $Q$-system.
\item[Choice] of the algorithm with the best height in these classes.
\item[Development] of a $Q$-effective program for the selected algorithm.
\end{description}

\section{Conclusion}

Now we can come to the following conclusions regarding the presented research results based on the concept of a $Q$-determinant.
\begin{enumerate}
\item
Detection and evaluation of the existing parallelism of any numerical algorithm is possible.
\item
Defining a admissible way for executing this parallelism is also possible.
\end{enumerate}
Finally, we can argue that there is a real opportunity to obtain the effective implementations of numerical algorithms for real parallel computing systems, as well as to increase the efficiency of the implementation of methods and algorithmic problems using numerical algorithms.

\section*{Acknowlegments}

The reported study was funded by RFBR according to the research project  no. 17-07-00865 a. 
The work was supported by Act 211 Government of the Russian Federation, contract no. 02.A03.21.0011.

\end{document}

\bibliographystyle{plain}
\bibliography{AleevaAleevV1}

\end{document}